\documentclass[aps,prd,notitlepage,nofootinbib,superscriptaddress,groupedaddress]{revtex4-2}

\usepackage[latin9]{inputenc}
\usepackage{mathrsfs}
\usepackage{bm}
\usepackage{amsmath}
\usepackage{amsthm}
\usepackage{amssymb}
\usepackage{graphicx}
\usepackage{esint}
\usepackage{hyperref}
\usepackage{array}[=2016-10-06]
\usepackage{tabularx}

\theoremstyle{plain}

\theoremstyle{plain}

\ifx\proof\undefined
\newenvironment{proof}[1][\protect\proofname]{\par
	\normalfont\topsep6\p@\@plus6\p@\relax
	\trivlist
	\itemindent\parindent
	\item[\hskip\labelsep\scshape #1]\ignorespaces
}{
	\endtrivlist\@endpefalse
}
\providecommand{\proofname}{Proof}
\fi

\usepackage[table ]{ xcolor}

\usepackage{tikz}
\usepackage{pgfplots}
\pgfplotsset{compat=1.18}

\usepackage{amssymb}
\usepackage{graphicx,color}

\usepackage{amsmath,amsthm}

\usepackage{mathrsfs}
\usepackage{subcaption}

\usepackage{bm}

\def\beq{\begin{equation}}
\def\eeq{\end{equation}}
\def\bi{\begin{itemize}}
\def\ei{\end{itemize}}
	\def\ba{\begin{array}}
	\def\ea{\end{array}}
	\def\bfig{\begin{figure}}
	\def\efig{\end{figure}}

	\def\C{\mathbb{C}}
	\def\R{\mathbb{R}}

	\newtheorem{theorem}{Theorem}[section]

	\newtheorem{lemma}[theorem]{Lemma}

	\def\be{\begin{eqnarray}}
	\def\ee{\end{eqnarray}}

	\newcommand{\ca}{\mathcal A}

	\newcommand{\cc}{\mathcal C}

	\newcommand{\cf}{\mathcal F}
	
	\newcommand{\ch}{\mathcal H}

	\newcommand{\cl}{\mathcal L}
	\newcommand{\cm}{\mathcal M}
	\newcommand{\cn}{\mathcal N}
	\newcommand{\co}{\mathcal O}

	\newcommand{\cs}{\mathcal S}
	\newcommand{\ct}{\mathcal T}
	\newcommand{\cu}{\mathcal U}
	
	\newcommand{\cw}{\mathcal W}

	  \newcommand{\Fd}{\mathfrak{D}}

	  \newcommand{\Fh}{\mathfrak{H}}

	  \newcommand{\Fp}{\mathfrak{P}}

	\newcommand{\ft}{\mathfrak{t}}  \newcommand{\Ft}{\mathfrak{T}}
	  \newcommand{\Fu}{\mathfrak{U}}

	\renewcommand{\a}{\alpha}
	\renewcommand{\b}{\beta}
	\newcommand{\g}{\gamma}

	\newcommand{\sig}{\sigma}
	\newcommand{\Sig}{\Sigma}
	\renewcommand{\l}{\lambda}
	
	\renewcommand{\o}{\omega}
	\renewcommand{\O}{\Omega}
	\renewcommand{\t}{\tau}

	\newcommand{\rmd}{\mathrm d}

	\newcommand{\lt}{\left}
	\newcommand{\rt}{\right}

	\newcommand{\lag}{\left\langle}
	\newcommand{\rag}{\right\rangle}

	\newcommand{\rmi}{\mathrm{i}}

	\newcommand{\tr}{\mathrm{Tr}}
	\newcommand{\Tr}{\mathrm{Tr}}

\newcommand{\wt}{\widetilde}
\newcommand{\wh}{\widehat}

\begin{document}

\title{Semiclassical state on algebra of observables for de Sitter space and time reversal symmetry}

\author{Muxin Han$^{1,2}$}
\email{hanm(At)fau.edu}
\author{Yichen Hu$^{1}$}
\email{huy(At)fau.edu}
\affiliation{1. Department of Physics, Florida Atlantic University, 777 Glades Road, Boca Raton, FL 33431, USA}
\affiliation{2. Institut f\"ur Quantengravitation, Universit\"at Erlangen-N\"urnberg, Staudtstr. 7/B2, 91058 Erlangen, Germany}

\begin{abstract}
In de Sitter static patch holography, coupling the bulk quantum fields to an observer's dynamical clock yields a type II$_1$ von Neumann algebra $\cm_+$ of gauge invariant observables. However, the unique tracial state on this algebra fails to reproduce semiclassical bulk correlation functions. In this paper, we resolve this issue by introducing a family of heavy-clock states, representing an observer reference frame with heavy inertia.
We prove that in the heavy-clock limit, $n$-point correlation functions of gauge-invariant dressed observables reproduce the corresponding bulk quantum field correlation functions on de Sitter spacetime. 
Furthermore, we investigate a generalized observer Hamiltonian $\widehat{H}_{\rm obs}=|\widehat{\pi}|$ that incorporates both forward-running ($\pi\geq 0$) and backward-running ($\pi <0$) clocks. The resulting von Neumann algebra $\mathcal{M}$ of gauge-invariant observables remains a type~II$_1$ factor and is isomorphic to $\cm_+\otimes M_2(\C)$, where $M_2(\mathbb{C})$ is the algebra of $2\times2$ complex matrices associated with transforming between the two clock orientations.
Time-reversal symmetry can be implemented in the setting with the generalized observer Hamiltonian, and it is realized as an anti-unitary operator that exchanges the two clock orientations.
A type~II$_1$ von Neumann algebra $\mathcal{F}$ is constructed to collect gauge-invariant observables that are also invariant under time-reversal. 
We generalize the heavy-clock state construction to both $\mathcal{M}$ and $\mathcal{F}$. In both settings, we show that $n$-point correlation functions calculated in heavy-clock states match the corresponding semiclassical bulk correlation functions.
Finally, by coupling Einstein gravity to the Brown--Kucha\v{r} dust field, we derive the observer Hamiltonian $\widehat{H}_{\rm obs}=|\widehat{\pi}|$ from a perturbative reduction along the observer's worldline in the $G \to 0$ limit. This provides a field-theoretic foundation for the observer model developed in this work.
\end{abstract}

\maketitle

\tableofcontents

\section{Introduction}

De Sitter (dS) holography broadly seeks a holographic description of quantum gravity in an accelerating universe with a positive cosmological constant. Within this program, the holography in static patch aims to describe quantum gravity in the universe accessible to a single observer, whose causal diamond is bounded by a cosmological event horizon. Unlike Anti-de Sitter (AdS) spacetime, which has a timelike conformal boundary, de Sitter spacetime has spacelike conformal boundaries at future and past infinity and lacks spatial infinity. 
The thermodynamic properties of the cosmological horizon make the formulation of a holographic description for the static patch compelling:
Gibbons and Hawking~\cite{IntroGibbonsHawking1977} demonstrated that the cosmological horizon emits thermal radiation at temperature $T_{\mathrm{dS}}=1/(2\pi L)$ and carries an entropy $S_{\mathrm{dS}}=A/(4G_N)$, where $L$ is the de Sitter radius and $A$ is the horizon area. The area-law entropy of the horizon suggests that there should be a holographic description for the static patch \cite{IntroBousso2002}. 
A challenge in relating the holographic description to the semiclassical bulk theory arises because, in quantum field theory on curved spacetime, the local algebra of static-patch observables forms a type III$_1$ von Neumann factor that lacks both a trace and a finite entropy.
To resolve this, Chandrasekaran, Longo, Penington, and Witten (CLPW)~\cite{Chandrasekaran:2022cip}, building on Witten's crossed-product construction~\cite{Witten:2021unn}, included an observer equipped with a dynamical quantum clock; the crossed product of the bulk algebra by the modular flow yields a type II$_1$ factor whose unique normalized trace realizes the Gibbons--Hawking entropy of the de Sitter horizon.

In the static patch of $D$-dimensional dS spacetime (FIG.\ref{static_patch}), the dS metric takes the following form
\be
\rmd s^2=-\lt(1-\frac{r^2}{L^2}\rt)\rmd t^2+\lt(1-\frac{r^2}{L^2}\rt)^{-1}\rmd r^2+r^2\rmd\O_{D-2}^2.\label{dSmetric}
\ee
Within the static patch, an observer follows a worldline denoted by $\g$ and has the horizon as the boundary of the patch $r=L$. 
We might envision the observer aligning with a ``stretched horizon'', a surface located at a short distance from the horizon \cite{Susskind:2023rxm}. However, to extend the generality of the scenario, we consider a general observer worldline completely within the static patch.

\begin{figure}[h]
\centering
\includegraphics[width=0.2\textwidth]{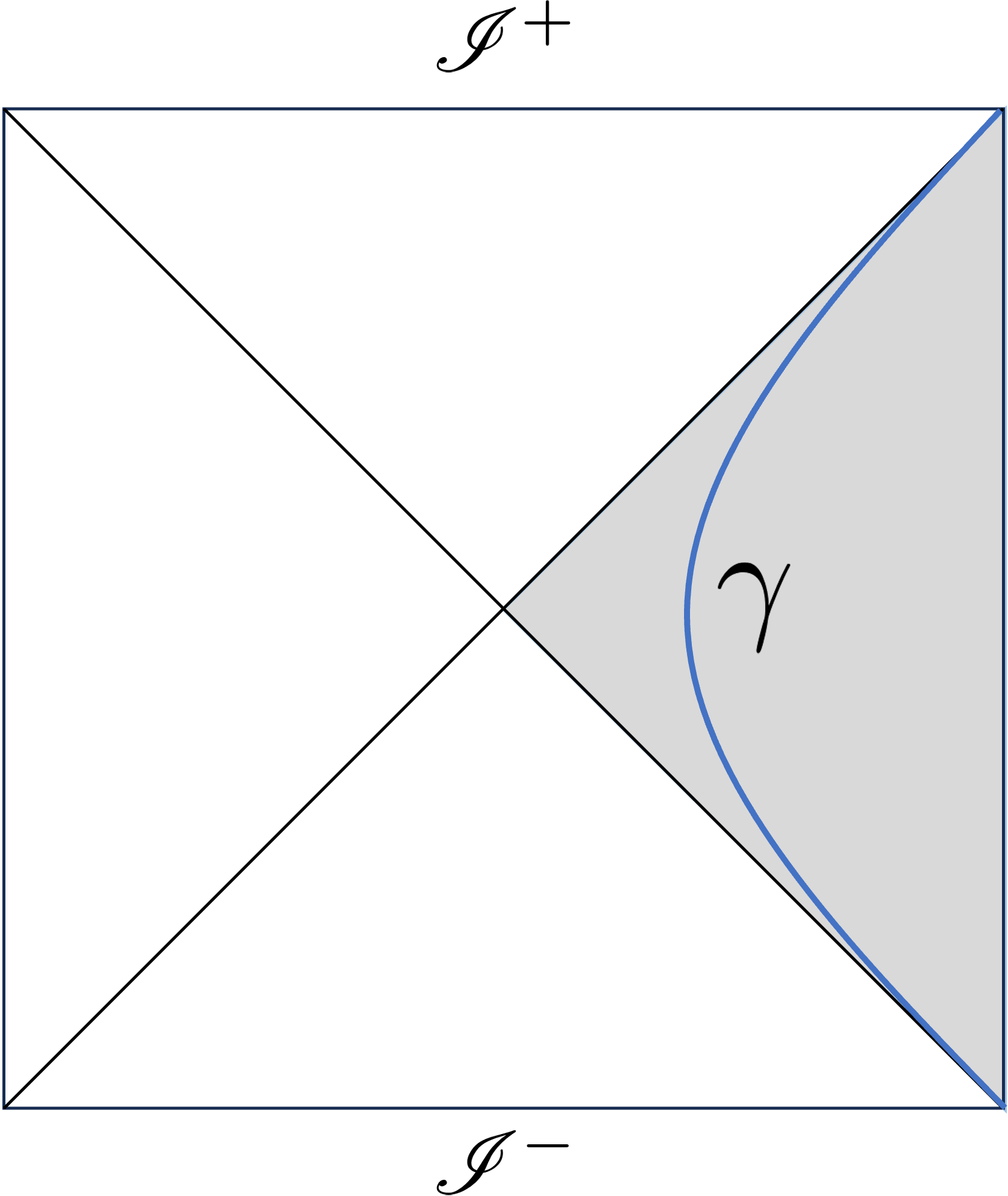}
\caption{Penrose diagram of de Sitter spacetime: The gray region represents the static patch. The blue curve $\gamma$ depicts the worldline of an observer.}
\label{static_patch}
\end{figure}
 
We denote by $\ca_0$ the type III$_1$ von Neumann algebra of local field operators within the static patch, represented on the Hilbert space $\ch_0$ of quantum fields.
The Killing vector $\partial/\partial t$ defines a Hamiltonian $\wh H$ on the Hilbert space $\ch_0$ and generates an abelian automorphism group on $\ca_0$
\be
\sigma_t(\co) = e^{it\wh H} \co e^{-it\wh H},\qquad \forall\, \co\in\ca_0\ .
\ee
This automorphism is modular: there exists a cyclic and separating state $\Psi_0$ of $\ca_0$, such that $\Delta=e^{-\wh H}$ is the corresponding modular operator. Consequently, correlation functions in the static patch, evaluated in the state $\Psi_0$, satisfy the KMS condition \cite{Chandrasekaran:2022cip}. This is analogous to the Bisognano--Wichmann theorem for the Rindler wedge.
Note that the operator $\wh H$ is not positive, as the Killing vector $\partial/\partial t$ becomes past-directed in the complementary static patch.

The static patch holography proposes the duality between the physical theory along $\g$ and the bulk theory in the static patch \cite{Susskind:2023rxm,Chandrasekaran:2022cip}. To confirm such a duality, one needs a dictionary that relates the bulk quantum fields to the holographic degrees of freedom living on $\g$. 
One might expect that taking the limit where a bulk operator approaches the observer's worldline should provide the dictionary. However, this is not entirely accurate, as the holographic theory contains the observer's own degrees of freedom in addition to $\ca_0$.
More importantly, when gravitational fluctuations are considered in the bulk theory, the Hamiltonian should be understood as the Hamiltonian constraint. 
Following the CLPW proposal, the modular Hamiltonian $\wh H$ is coupled to the observer's Hamiltonian $\wh H_{\rm obs}$, and the total Hamiltonian $\wh{\bf H} = \wh H + \wh H_{\rm obs}$ is the Hamiltonian constraint in the $G\to 0$ limit, which generates gauge transformations. This imposes the requirement that all physical operators must be invariant under the gauge transformations.
The resulting operator algebra is a type II$_1$ von Neumann algebra, which has a unique trace up to scaling. This trace defines a state on which one can compute the expectation value of any gauge invariant operator on the worldline. It is conjectured that this tracial state should correspond to the empty dS spacetime. 

In the limit of very large horizon area ${A_{\rm dS}}/{(4 G)}\gg 1$, the static patch can be described by the semiclassical equations of motion. In particular, the approximation of quantum field theory in a fixed dS background is accurate as long as the energies involved are not large enough to create significant back reaction on the geometry. A puzzle arises, as pointed out by Susskind in \cite{Susskind:2023rxm,Susskind:2026fjc}, regarding the semiclassical behavior of correlation functions on the tracial state. Consider two self-adjoint bulk field operators, $\co_1$ and $\co_2$, localized on the observer's worldline $\gamma$. Their bulk correlation function is $\langle\Psi_0| \co_1\co_2| \Psi_0\rangle$. The imaginary part of this bulk correlation function is proportional to the expectation value of their commutator, i.e., $ \langle \Psi_0| [\co_1,\co_2]| \Psi_0\rangle $, which is generally nonzero. In contrast, when the same correlation function is computed on the worldline, it is given by the trace $\Tr(\co'_1\co'_2)$,
where the self-adjoint operators $\co'_1$ and $\co'_2$ denote the gauge-invariant dressings of $\co_1$ and $\co_2$,
then the imaginary part always vanishes by the cyclicity of trace. 
This discrepancy indicates that the tracial state fails to capture the correct semiclassical behavior. 
Indeed, as discussed by CLPW in the context of entropy computation \cite{Chandrasekaran:2022cip}, the tracial state cannot be interpreted as a semiclassical state of the observer. 
Once the observer's coupling is included, the semiclassical properties of the observer's state influence the semiclassical behavior of the entire system. This gives rise to the actual tension underlying the problem.

One of the main purposes of this paper is to address the issue on the semiclassical limit of the correlation function for gauge-invariant dressed operators. 
We demonstrate that for a wide class of heavy-clock states $\o_\a$, the semiclassical limit of arbitrary $n$-point correlation functions of gauge-invariant dressed operators $A_{1,+},\cdots, A_{n,+}$ correctly reproduces the $n$-point functions of bulk quantum field operators $a_1,\cdots,a_n$ in de Sitter spacetime:
\be
\lim_{\alpha\to0}\omega_\alpha(A_{1,+}\cdots A_{n,+})
=\left\langle\Psi_0\middle|a_1\cdots a_n\middle|\Psi_0\right\rangle.
\label{intro:heavy-clock-correspondence}
\ee
Here the heavy-clock state generalizes the semiclassical states considered in \cite{Chandrasekaran:2022cip} by relaxing the requirements of differentiability. The limit $\a\to0$ is often referred to as the heavy-clock limit throughout the paper.

Physically, the requirement that the clock be heavy reflects the fact that an ideal classical reference frame must have heavy inertia so that the perturbation from bulk fluctuations is negligible. In a fully relational quantum theory, reference frames are physical, dynamical degrees of freedom. When bulk operators act, they induce transitions of finite modular energy $E$, which, by the total constraint $\wh{\bf H} = \wh H + \wh H_{\rm obs} = 0$, must be compensated by a recoil of the clock momentum $\Delta H_{\rm obs} \sim -E$. If the clock carries only finite energy, this recoil substantially perturbs its state and entangles the reference frame with the quantum fields, causing relational measurements to deviate from the fixed-background prediction. In contrast, in the heavy-clock limit where the typical clock energy scales as $\langle\wh H_{\rm obs}\rangle \sim 1/\alpha \to \infty$, then the ratio $\Delta H_{\rm obs}/\langle \wh H_{\rm obs}\rangle \sim \alpha E \to 0$ vanishes. The clock thus acts as an unperturbed, macroscopic reference frame, allowing the relational observables to recover the semiclassical dynamics of quantum field theory on a fixed de Sitter background. 
The backreaction from the clock to spacetime geometry is neglected because $\wh{\bf H}$ is the constraint in the $G\to0$ limit (see Section \ref{sec:dust clock fields} for an explanation). 

Let the clock canonical variables be $\wh\t$ and $\wh\pi$. In the CLPW proposal, the observer's Hamiltonian is taken as $\wh H_{\rm obs} = \wh\pi$, and the Hilbert space is restricted to the sector with $\wh\pi \ge 0$. The semiclassical limit \eqref{intro:heavy-clock-correspondence} is derived in this setup. In addition, we study a generalized observer's Hamiltonian $\wh{H}_{\mathrm{obs}}=|\wh\pi|$, accommodating both $\pi \geq 0$ and $\pi<0$ in the Hilbert space. The sectors of positive and negative $\pi$ correspond to forward- and backward-running clock orientations. 
The von Neumann algebra $\cm$ of gauge-invariant operators remains a type II$_1$ factor, and it is related to the CLPW algebra $\cm_+$ by $\mathcal{M} \cong \mathcal{M}_+ \bar{\otimes} M_2(\mathbb{C})$, where $M_2(\mathbb{C})$ is the algebra of $2\times2$ complex matrices that captures the degrees of freedom associated with transforming between the two clock orientations. 
We extend the result \eqref{intro:heavy-clock-correspondence} by constructing heavy-clock states on the full algebra $\mathcal{M}$, and show that in the heavy-clock limit, $n$-point functions of gauge-invariant operators in $\cm$ reproduce the corresponding bulk correlation functions.

The generalization to $\wh{H}_{\mathrm{obs}}=|\wh\pi|$ is inspired by Susskind's proposal \cite{Susskind:2026fjc}, which incorporates time-reversal symmetry directly into the theory. 
In the above framework with the algebra $\cm$, and relating to the recent developments in quantum reference frame transformations (see, e.g., \cite{Thiemann:2026vaj,Chen:2026kui}), time-reversal symmetry is interpreted as a global symmetry that relates two distinct quantum reference frames.
However, the status of time reversal in quantum gravity remains debated: there are proposals to treat time reversal as gauge symmetries, while other treatments distinguish time reversal from gauge transformations \cite{Chen:2014xhe,IntroHarlowNumasawa,IntroWittenTimeReversal}.
Susskind proposes that time reversal is gauged but hidden by a Higgs-like mechanism, where clock orientations play the role of an Ising-like variable, and persistent correlations with bulk fields pick out a preferred time orientation \cite{Susskind:2026fjc,IntroSusskind2603}. 
In this paper, while we do not assume time reversal is gauged, we nevertheless investigate the construction of a von Neumann algebra of dressed operators that are invariant under both gauge transformations by $\wh{\bf H}$ and time reversal. As a result, we construct a type-II$_1$ von Neumann algebra $\cf$ of strictly time-reversal-invariant observables within the gauge invariant dressed operators, and formulate a representation which preserves time-reversal invariance for operator products. By the heavy-clock state, the result \eqref{intro:heavy-clock-correspondence} generalizes to the operator in $\cf$: the $n$-point function reproducing the correct semiclassical bulk correlation function. It demonstrates that the semiclassical bulk recovery requires neither gauging nor spontaneous breaking of time reversal.

Another key motivation for considering $\wh{H}_{\mathrm{obs}}=|\wh\pi|$ is that this form of the observer's Hamiltonian arises naturally from the field theory of gravity coupled to matter. Indeed, we generalize Witten's idealized point-like worldline clock to a dynamical field-theoretic quantum reference frame via Brown--Kucha\v{r} pressureless dust coupled to Einstein gravity~\cite{BrownKuchar1995,Giesel2010MGIPT}. 
By performing a perturbative worldline reduction inside a comoving dust tube, we show that the abelianized gravitational Hamiltonian constraint naturally results in $\wh{H}_{\mathrm{obs}}=|\wh\pi|$. 
Reversing the time orientation of spacetime induces the transformation $\pi \to -\pi$, thereby exchanging the two clock orientations.
This derivation provides a field-theoretic foundation for the observer model with two clock orientations.

As outlined above, there are three closely related operator-algebraic frameworks explored in this work: (1) the original CLPW construction featuring a positive-energy clock (with $\wh H_{\rm obs} = \wh\pi$), (2) the generalized formulation that incorporates both clock orientations via $\wh H_{\rm obs} = |\wh\pi|$, and (3) the construction of strictly time-reversal-invariant operators. For comparison and to highlight their respective properties, the main aspects of these three constructions are presented side by side in Table~\ref{tab:my_table}. 
In all three scenarios, the $n$-point functions evaluated in heavy-clock states correctly reproduce the semiclassical bulk correlation functions.

\begin{table}[!h]
\centering
\small
\renewcommand{\arraystretch}{1.2}
\setlength{\tabcolsep}{3pt}
\makeatletter
\let\insert@pcolumn\insert@column
\makeatother
\begin{tabularx}{\textwidth}{@{}>{\raggedright\arraybackslash}p{0.15\textwidth}>{\raggedright\arraybackslash}X>{\raggedright\arraybackslash}X>{\raggedright\arraybackslash}X@{}}
\hline
 & Positive-clock (CLPW) & Two clock orientations & Time-reversal-invariant operators \\
\hline
Clock sectors
 & $\pi\geq0$ 
 & $\pi\geq 0$ and $\pi<0$ 
 & $\pi\geq 0$ and $\pi<0$ \\
\hline
Clock Hamiltonian
 & $\wh H_{\rm obs}=\wh\pi$
 & $\wh H_{\rm obs}=|\wh\pi|$
 & $\wh H_{\rm obs}=|\wh\pi|$ \\
\hline
Operator algebra
 & type II$_1$ factor $\cm_+$
 & type II$_1$ factor $\cm\cong\cm_+\bar\otimes M_2(\C)$
 & type II$_1$ factor $\mathcal F=j(\cm_+)\cong\cm_+$ \\
\hline
Heavy-clock state
 & $\omega_\alpha:\cm_+\to\C$
 & $\omega_{\alpha,\rho}=\omega_\alpha\otimes \mathrm{tr}(\rho\ \cdot)$
 & $\omega_{\alpha,\rho}$ \\
\hline
\end{tabularx}
\caption{Summary of the three related constructions: The map $j$ is an isomorphism from $\cm_+$ to $\cf\subset \cm_+\oplus\cm_-^{\rm op}$ ($\cm_-$ is the algebra of gauge-invariant operators defined on the $\pi<0$ sector, and $\cm_-^{\rm op}$ is the algebra obtained by reversing the multiplication in $\cm_-$). The operators in the image of $j$ are time-reversal invariant.}
\label{tab:my_table}
\end{table}

The paper is organized as follows. Section~\ref{sec:crossed-product-algebra} reviews the de Sitter static patch observer algebra and its crossed product with a positive-energy clock. Section~\ref{sec:two-point-crossed-product} analyzes two-point correlation functions of dressed observables and introduces the heavy-clock state. Section~\ref{sec:Holographic correspondence of npoint correlation functions} establishes the heavy-clock limit for general $n$-point correlation functions, recovering the semiclassical bulk correlation function. Section~\ref{sec:local-observer-time-reversal} incorporates both clock orientations via the clock Hamiltonian $\wh H_{\rm obs}=|\wh\pi|$ and shows that the gauge-invariant observable algebra $\cm$ is a type~II$_1$ factor isomorphic to $\cm_+\bar\otimes M_2(\C)$. Section~\ref{sec:folded-bulk-boundary} constructs relationally dressed observables for both orientations and proves that linearity and self-adjointness uniquely select the dressing that reproduces bulk correlation functions. Section~\ref{sec:Time reversal symmetry} formulates the anti-unitary time-reversal symmetry exchanging the two clock orientations and proves the invariance of the trace on $\cm$. Section~\ref{sec:time-reversal-invariant-operators} constructs the type~II$_1$ algebra $\cf$ of strictly time-reversal-invariant observables and formulates an operator product preserving time-reversal invariance. Section~\ref{sec:dust clock fields} derives the observer Hamiltonian $\wh{H}_{\rm obs}=|\wh\pi|$ from Einstein gravity coupled to Brown--Kucha\v{r} dust in the $G\to 0$ limit. Finally, technical proofs and background on duality, traces, and anti-automorphisms are collected in Appendices~\ref{Takesaki duality}--\ref{anti-automorphism}.

\section{Observer, quantum clock, and crossed-product algebra}\label{sec:crossed-product-algebra}

The de Sitter static patch can be described from the viewpoint of a single observer equipped with a quantum clock \cite{Chandrasekaran:2022cip}. 
Static-patch time translations act as a modular flow on the observer's algebra. The appropriate algebra of relational observables is therefore not the original algebra $\mathcal{A}_0$ itself, but its crossed product by this flow. The crossed product incorporates the observer's clock as a quantum reference frame and converts the type III von Neumann algebra into a type II algebra with a trace. We review this construction here.

Let $\mathcal{A}_0$ be the von Neumann algebra of local field operators in the static patch. The algebra $\mathcal{A}_0$ is the operator algebra of the system being observed by the observer $\g$. The Hilbert space carrying $\ca_0$ is $\mathcal{H}_0$. Let $\Psi_0\in\ch_0$ be the cyclic and separating state. The corresponding modular Hamiltonian $\widehat H=i\partial/\partial t$ associated with $(\mathcal{A}_0,\Psi_0)$ generates the modular automorphism group $\sigma_t(a)=e^{it\widehat H}a e^{-it\widehat H}$, $a\in\mathcal{A}_0$.
In the following, we generally refer to $\ch_0$ and the associated algebra $\ca_0$ as the ``system'' being observed by the observer.

We couple this system to the observer, whose Hilbert space is $\ch_{\rm obs}=L^2(\R)$. The Hilbert space $\ch_{\rm obs}$ carries canonical operators $(\widehat\tau,\widehat\pi)$ satisfying $[\widehat\tau,\widehat\pi]=i$. The position operator $\wh\t$ is the quantization of the clock time carried by the observer.
Intuitively, we are inclined to interpret the momentum operator $\wh \pi$ as the observer's Hamiltonian, because it is canonically conjugate to the observer's time variable.
However, $\wh \pi$ is not bounded from below and thus does not define a proper Hamiltonian. There are two possible prescriptions so that the resulting Hamiltonian is bounded from below:
\begin{enumerate}

\item Following the approach of \cite{Chandrasekaran:2022cip}, we restrict the spectrum of $\wh\pi$ to be positive by introducing the spectral projection operator $\Pi_+$. This projects the observer's Hilbert space $\ch_{\rm obs}$ onto the subspace $\Pi_+\ch_{\rm obs}\cong L^2(\R_{\geq0})$ with positive momentum. We will implement this prescription momentarily in this section, and also in Section \ref{sec:two-point-crossed-product}.
\item Alternatively, following the proposal in \cite{Susskind:2026fjc}, the observer's Hamiltonian can be taken as the absolute value operator $|\wh \pi|$. This choice is also closely related to the physical scenario in which gravity and matter are coupled to a dust clock field; see Section \ref{sec:dust clock fields} for details.
In addition, it also allows us to implement the time-reversal symmetry to the model. 
We will discuss this approach in Section \ref{sec:local-observer-time-reversal} and subsequent sections.

\end{enumerate}

To begin, let us temporarily set the Hamiltonian to $\wh \pi$; we will address the projection to positive spectrum in the next step.
The full Hilbert space is $\ch_{\rm tot}:=\mathcal{H}_0\otimes \ch_{\rm obs}\cong L^2(\mathbb R,\mathcal{H}_0)$.
The relational description is implemented by the Hamiltonian constraint, which sums the Hamiltonians of both the system and observer:
\be
\wh{\bf H}=\widehat H\otimes \mathbb I_{\rm obs}+\mathbb I_0\otimes \widehat\pi.\label{HamConstr}
\ee
acting on $\ch_{\rm tot}$. Here $\mathbb I_{\rm obs}$ and $\mathbb I_0$ are identity operator on $\ch_{\rm obs}$ and $\ch_0$ respectively.
Analogous to the situation in canonical quantum gravity, the unitary group generated by the Hamiltonian constraint $\wh{\bf H}$ implements gauge transformations. Indeed, the construction here is a toy model for reduced phase space quantization as developed in non-perturbative quantum gravity (see, e.g., \cite{Giesel2010MGIPT,Domagala:2010bm,Han:2009aw}), where the role of the observer $\g$ is played more generally by clock fields. 
In the context of quantum gravity coupled to clock fields, the ``abelianized Hamiltonian constraint'' retains the form given in \eqref{HamConstr}, with $\wh\pi$ representing the energy density associated with the clock fields.

Physical observables are operators that are invariant under the gauge transformations generated by the Hamiltonian constraint. In our model, the algebra $\mathcal{N}$ of gauge invariant operators is the commutant of the gauge flow $e^{i t\wh{\bf H}}$ (for all $t\in\R$) inside the total operator algebra $\ca\equiv \mathcal{A}_0\bar\otimes\mathcal{B}(L^2(\mathbb R))$, i.e.
\be
\mathcal{N}=
\lt\{\mathcal{O}\in\ca\;\middle|\;
e^{i t\wh{\mathbf H}}\mathcal{O}e^{-i t\wh{\mathbf H}}=\mathcal{O},
\;\forall t\in\mathbb{R}\rt\}.
\ee
This commutant is equivalent to the crossed-product von Neumann algebra
\be
\mathcal{N}\cong\mathcal{A}_0\rtimes_\sigma\mathbb R,\label{equivalencecalN}
\ee
where $\mathbb R$ acts by the modular flow $\sigma_t$. This relation is known as the Takesaki duality (see e.g. \cite{takesaki1973duality,takesaki2003theory2,Witten:2021unn}, and see also appendix \ref{Takesaki duality} for a brief review).

The algebra $\mathcal{N}$ is generated by two types of operators. The first is the clock translation operator
\be
\lambda_t=e^{-it\widehat\pi},\quad \forall t\in\R,\qquad (\lambda_t\xi)(\tau)=\xi(\tau-t),\qquad \forall \xi\in \ch_{\rm tot},
\ee
where we have represented $\xi$ as a function in $L^2(\R,\ch_0)$. 
The second type of generators is the translation of $a\in\mathcal{A}_0$ using quantum clock $\wh \t$,
\be
\pi_0(a)=e^{-i\widehat\tau\widehat H}a e^{i\widehat\tau\widehat H},\qquad
(\pi_0(a)\xi)(\tau)=\sigma_{-\tau}(a)\xi(\tau).
\ee
Both $\l_t$ and $\pi_0(a)$ commute with $\mathbf{\widehat H}$.

The operators $\pi_0(a)\in\cn$ promoted from $a\in\ca_0$ are called relational observables \cite{Rovelli:2001bz,Dittrich2007,Thiemann:2004wk}: the time dependence of the operator is in terms of the observer's quantum clock, which itself is dynamical and forms part of the quantum system. 
The evolution of the operator $a$ is described relationally in terms of the quantum dynamics of the clock variable $\tau$.

The equivalence \eqref{equivalencecalN} is realized via the conjugation by $e^{i\widehat\tau\widehat H}$: 
The crossed product algebra $\mathcal{A}_0\rtimes_\sigma\mathbb R$ is generated by $a\in\mathcal{A}_0$ and $e^{is(\widehat H-\widehat\pi)}$, $s\in\R$, which relates to $\pi_0(a)$ and $\l_a$ via
\be
\pi_0(a)\lambda_s=e^{-i\widehat\tau\widehat H}\left[a\,e^{is(\widehat H-\widehat\pi)}\right]e^{i\widehat\tau\widehat H}.
\ee

A dense set of elements in $\cn$ (in weak operator topology) is spanned by operators of the form
\be
\widehat a_F=\int_{-\infty}^\infty\rmd s\, F(s)\,\pi_0(a)\lambda_s=e^{-i\widehat\tau\widehat H}\int_{-\infty}^\infty\rmd s\, F(s)\left[a\,e^{is(\widehat H-\widehat\pi)}\right]e^{i\widehat\tau\widehat H}.
\ee
Here $F$ is a suitable test function on $\mathbb{R}$, for example a Schwartz function.
As the von Neumann algebra $\mathcal{A}_0$ is a type III$_1$ factor, the crossed product $\mathcal{N}\cong\mathcal{A}_0\rtimes_\sigma\mathbb R$ is a type II$_\infty$ factor. It carries a faithful, normal, semifinite trace $\mathrm{Tr}$, unique up to an overall normalization (see e.g. \cite{Sorce:2023fdx,Witten:2021unn,takesaki2003theory2}, see also Appendix \ref{Trace on cross-product algebra} for a short review). For the element $\widehat a_F$ above, the trace reads
\be
\mathrm{Tr}(\widehat a_F)
&=&\int_{-\infty}^{\infty}\rmd \pi\,e^{-\pi }\left\langle\Psi_0\left|\int\rmd s\,F(s)\,a e^{is(\widehat H-\pi)}\right|\Psi_0\right\rangle\nonumber\\
&=&\int_{-\infty}^{\infty}\rmd \pi\,e^{-\pi}\wt F(\pi)\left\langle\Psi_0\middle|a\middle|\Psi_0\right\rangle,\qquad \wt F(\pi)=\int\rmd s\,e^{-is\pi}F(s).
\ee
We have used that $\wh H|\Psi_0\rangle =0$, and $\pi$ denotes the spectral value of $\widehat\pi$. 
The trace is semifinite rather than finite because the full crossed product retains the entire clock-momentum line; in particular, the identity has infinite trace.

In accordance with prescription 1 above, we now restrict the type II$_\infty$ factor $\cn$ to the sector where the observer's energy is positive. This amounts to restricting the Hilbert space to $\Pi_+\ch_{\rm tot} = \ch_0 \otimes L^2(\mathbb{R}_{\geq 0})$, so that the spectrum of $\wh\pi$ is confined to $\mathbb{R}_{\geq 0}$.
Since $\widehat\pi$ strongly commutes with $\widehat H+\widehat\pi$, its spectral projections belong to $\mathcal{N}$. In particular,
\be
\Pi_+=\chi_{[0,\infty)}(\widehat\pi)\in\mathcal{N}.
\ee
where $\chi_{I}$ denotes the indicator function of the interval $I$.
The algebra of gauge invariant observables is given by projecting the crossed product algebra $\cn$ onto $\Pi_+\ch_{\rm tot}$
\be
\mathcal{M}_+=\Pi_+\mathcal{N}\Pi_+.\label{cmPlus}
\ee
The trace on $\cm_+$ is given by the restriction of $\Tr$. Let us project the operator $\wh a_F\in\cn$ into $\cm_+$ then evaluate the trace. Using the relation
\be
\Pi_+e^{-i\widehat\tau\widehat H}\left[a\,e^{is(\widehat H-\widehat\pi)}\right]e^{i\widehat\tau\widehat H}\Pi_+
=e^{-i\widehat\tau\widehat H}\left[\chi_{[0,\infty)}(\widehat\pi-\widehat H)a e^{is(\widehat H-\widehat\pi)}\chi_{[0,\infty)}(\widehat\pi-\widehat H)\right]e^{i\widehat\tau\widehat H},\label{exchangeRelationPiPlus}
\ee
The operators $\Pi_+\widehat a_F\Pi_+$ spans a dense subset in $\cm_+$\footnote{
	Every $A_+ \in \cm_+$ can be written as $A_+ = \Pi_+ \widehat{a} \Pi_+$ for some $\widehat{a} \in \cn$. Because the set $\{\wh a_F\}$ is dense in $\cn$, there is a net of such elements converging to $\widehat{a}$ in the weak operator topology: $\widehat{a} = \lim \wh a_F$. For any $\psi, \phi \in \ch_{\rm tot}$, it follows that $\langle \psi | \Pi_+ \wh a_F \Pi_+ | \phi \rangle \to \langle \psi | \Pi_+ \widehat{a} \Pi_+ | \phi \rangle = \langle \psi | A_+ | \phi \rangle$ in the limit. Thus, $A_+ = \lim \Pi_+ \wh a_F \Pi_+$ in the weak operator topology.
}.
The trace of $\Pi_+\widehat a_F\Pi_+$ is expressed as
\be
\mathrm{Tr}_+\left(\Pi_+\widehat a_F\Pi_+\right)
&=&\int_{-\infty}^{\infty}\rmd \pi\,e^{-\pi}\left\langle\Psi_0\left|\int\rmd s\,F(s)\,\chi_{[0,\infty)}(\pi-\widehat H)a e^{is(\widehat H-\pi)}\chi_{[0,\infty)}(\pi-\widehat H)\right|\Psi_0\right\rangle\nonumber\\
&=&\int_{0}^{\infty}\rmd \pi\,e^{-\pi}\wt F(\pi)\left\langle\Psi_0\middle|a\middle|\Psi_0\right\rangle.\label{PiaFPi}
\ee
The range of integral reduces to half real-line, so the trace becomes finite.
In particular, $\Pi_+$ is the identity of $\mathcal{M}_+$ and has finite trace,
\be
\mathrm{Tr}_+(\Pi_+)=\mathrm{Tr}(\Pi_+)=\int_{0}^{\infty}\rmd \pi\,e^{-\pi}=1.
\ee
This indicates that $\cm_+$ is a type II$_1$ factor.

Given a bulk local field operator $a\in\ca_0$, we dress $a$ by coupling to the quantum clock and define the corresponding ``dressed operator'' as 
\be
A_+:=\Pi_+\pi_0(a)\Pi_+\in\mathcal{M}_+, \quad \text{where}\quad \pi_0(a)=e^{-i\widehat\tau\widehat H}a e^{i\widehat\tau\widehat H}.\label{holographicMappingPlus}
\ee
Through this map, the bulk excitation is promoted to a gauge-invariant relational observable on $\g$.

For one-point functions and using trace on $\cm_+$, we have the following correspondence between expectation values of the dressed operators and bulk operators:
\be
\mathrm{Tr}_+(A_+)=\left\langle\Psi_0\middle|a\middle|\Psi_0\right\rangle .
\ee
The right-hand side represents the semiclassical expectation value of the bulk local operator in the static patch vacuum. 
At first glance, this formula might suggest the tracial state $\tr_+$ is a semiclassical state. 
However, this interpretation is not correct, because it cannot be extended to $n$-point functions for $n>1$, as we see in the next section. Moreover, the relation above is not unique to the maximally mixed state; it holds for any density state $\varrho(\widehat{\pi})\in\cm_+$ of the clock:
\be
\mathrm{Tr}_+\left(\varrho(\widehat{\pi})^{1/2}A_+\varrho(\widehat{\pi})^{1/2}\right)=\left\langle\Psi_0\middle|a\middle|\Psi_0\right\rangle,\label{1ptFunctionRelation}
\ee
provided the normalization condition $\mathrm{Tr}_+\left(\varrho(\widehat{\pi})\right)=\int_{0}^{\infty}\rmd \pi\, e^{-\pi} \varrho(\pi)=1$ is satisfied. Thus, the agreement of the one-point function in this construction does not uniquely determine the state on $\cm_+$.

\section{Correlation functions and heavy-clock state}\label{sec:two-point-crossed-product}

Our aim is to generalize the correspondence between dressed and bulk operators to $n$-point correlation functions with arbitrary $n\geq 1$. Let us first consider the two-point function. For two bulk local operators $a,b\in\mathcal{A}_0$, we use the dressing map \eqref{holographicMappingPlus} to define $A_+=\Pi_+\pi_0(a)\Pi_+$ and $B_+=\Pi_+\pi_0(b)\Pi_+$. We write $\chi_+(x)=\chi_{[0,\infty)}(x)$. Using the relation \eqref{exchangeRelationPiPlus},
\be
A_+=e^{-i\widehat\tau\widehat H}\chi_+(\widehat\pi-\widehat H)a\chi_+(\widehat\pi-\widehat H)e^{i\widehat\tau\widehat H},\qquad B_+=e^{-i\widehat\tau\widehat H}\chi_+(\widehat\pi-\widehat H)b\chi_+(\widehat\pi-\widehat H)e^{i\widehat\tau\widehat H},
\ee
Their product in $\mathcal{M}_+$ is
\be
A_+B_+=e^{-i\widehat\tau\widehat H}\chi_+(\widehat\pi-\widehat H)a\chi_+(\widehat\pi-\widehat H)b\chi_+(\widehat\pi-\widehat H)e^{i\widehat\tau\widehat H}.
\ee
The presence of the intermediate projector $\chi_+(\widehat\pi-\widehat H)$ between $a$ and $b$ plays a crucial role.
Let us consider the action of $A_+B_+$ on the state $|\Psi_0\rangle\otimes|\phi\rangle$ for any clock state $|\phi\rangle\in L^2(\R_{\geq 0})$. By $\wh H|\Psi_0\rangle=0$, the operator $b$ acts on $|\Psi_0\rangle\otimes\Pi_+|\phi\rangle$. However, the action of $b$ changes the system state $\Psi_0$ to an intermediate state with nontrivial modular energy $E$ (we use $E$ to denote the spectrum of $\wh H$). The constraint $\chi_+(\widehat\pi-\widehat H)$ between $a$ and $b$ relates this system energy $E$ to the clock energy $\pi$: in the intermediate channel, the condition for staying in the sector of positive clock energy is therefore no longer $\pi>0$ but $\pi-E>0$. In other words, if $E>0$, the clock must carry enough positive energy to absorb the energy $E$ transferred by the first operator $b$ before the second operator $a$ acts.

Our aim is to relate the expectation value of $A_+B_+$ to the bulk two-point correlation function $\left\langle\Psi_0\middle|a b\middle|\Psi_0\right\rangle$. By the relation \eqref{1ptFunctionRelation}, we have
\be
\left\langle\Psi_0\middle|a b\middle|\Psi_0\right\rangle=\mathrm{Tr}_+\left(\varrho(\widehat{\pi})^{1/2}\lt[\Pi_+\pi_0(ab)\Pi_+\rt]\varrho(\widehat{\pi})^{1/2}\right),\label{bulk2ptFunc}
\ee
with any density state $\varrho(\widehat{\pi})$. The operator insertion on the right-hand side is 
\be
\Pi_+\pi_0(ab)\Pi_+=e^{-i\widehat\tau\widehat H}\chi_+(\widehat\pi-\widehat H)ab\chi_+(\widehat\pi-\widehat H)e^{i\widehat\tau\widehat H}.
\ee
The key difference from $A_+B_+$ is that this operator lacks the intermediate projector between $a$ and $b$. 

In contrast, when evaluating the expectation value $\mathrm{Tr}_+\left(\varrho(\widehat{\pi})^{1/2}A_+B_+\varrho(\widehat{\pi})^{1/2}\right)$, the presence of the intermediate projector has a significant effect if the clock's energy is small: it excludes some positive-energy intermediate states, causing this expectation value to deviate from the bulk correlation function $\left\langle\Psi_0\middle|a b\middle|\Psi_0\right\rangle$. However, in the regime of a ``heavy'' clock, meaning the typical clock momentum is much larger than the modular energies associated with the operators, the condition $\chi_+(\pi-E)\simeq 1$ holds across the relevant spectrum. In this regime, the influence of the intermediate projection becomes negligible, so the expectation value of $A_+B_+$ can effectively reproduce the bulk correlation function.

More explicitly, the expectation value $\mathrm{Tr}_+\left(\varrho(\widehat{\pi})^{1/2}A_+B_+\varrho(\widehat{\pi})^{1/2}\right)$ can be expressed by inserting a resolution of the identity in terms of the spectrum of $\widehat H$, with $\widehat H|E\rangle=E|E\rangle$ \footnote{
	Interchanging the order of integral can be justified as follows: We define the spectral density $m_{a,b}(E):=\left\langle\Psi_0\middle|a\middle|E\right\rangle\left\langle E\middle|b\middle|\Psi_0\right\rangle $. We have the following bound:
	\be 
	\int\rmd E\,|m_{a,b}(E)|&\leq&\left(\int\rmd E\,\left|\left\langle\Psi_0\middle|a\middle|E\right\rangle\right|^2\right)^{1/2}\left(\int\rmd E\,\left|\left\langle E\middle|b\middle|\Psi_0\right\rangle\right|^2\right)^{1/2}=\|a^*\Psi_0\|\,\|b\Psi_0\|<\infty ,\label{schwarzinequalityab}
	\ee 
	since $a,b\in\mathcal A_0$ are bounded. Since $p_\alpha(\pi)=e^{-\pi}\varrho(\pi)$ is normalized, we have the absolute convergence 
	\be\int_0^\infty\rmd\pi\,p_\alpha(\pi)\int\rmd E\,\chi_+(\pi-E)|m_{a,b}(E)|\leq\int\rmd E\,|m_{a,b}(E)|<\infty .
	\ee
	Fubini's theorem therefore justifies interchanging the $\pi$ and $E$ integrals,
}: 

\be
\mathrm{Tr}_+\lt(\varrho(\widehat{\pi})^{1/2}A_+B_+\varrho(\widehat{\pi})^{1/2}\rt)=\int\rmd E\,K_{\varrho}(E)\left\langle\Psi_0\middle|a\middle|E\right\rangle\left\langle E\middle|b\middle|\Psi_0\right\rangle,\label{Tr+rhoABrho}
\ee
where
\be
K_{\varrho}(E)=\int_0^\infty\rmd\pi\,e^{-\pi}\chi_+(\pi-E)\varrho({\pi})
\ee
Here and below, $\rmd E$ denotes the spectral measure of $\widehat H$. 

The bulk two-point function would be reproduced if $K_{\varrho}(E)=1$. In the case of tracial state where $\varrho=1$, we obtain
\be
K_{\varrho=1}=\int_0^\infty\rmd\pi\,e^{-\pi}\chi_+(\pi-E)=\begin{cases}
 e^{-E},&E>0,\\
 1,&E<0,
\end{cases}
\ee
In the tracial state, the clock energy distribution is $e^{-\pi}$, so states with sufficiently large clock energy are exponentially rare. It results in the suppression of states with positive intermediate modular energies. This is the precise obstruction to obtaining the semiclassical bulk correlator directly from the trace.

The remedy is to use a state that describe a heavy clock: Let us first introduce a family of functions on $\R_{\geq 0}$
\be
f_\alpha(\pi)=\alpha e^{(1-\alpha)\pi},\qquad 0<\alpha\leq 1,
\ee
and introduce a family of clock density states given by $\varrho(\widehat{\pi}) = f_\alpha(\widehat{\pi})$. This motivates us to define the linear functional $\omega_\alpha$ on $\mathcal{M}_+$ as the following: 
Consider a general element of the form $A_{F,+} = \Pi_+\widehat a_F\Pi_+$,
where $\widehat a_F = e^{-i\widehat\tau\widehat H}\int \rmd s\, F(s)\left[a\,e^{is(\widehat H-\widehat\pi)}\right]e^{i\widehat\tau\widehat H}$. 
\be
\omega_\alpha\lt( A_{F,+}\rt):=\mathrm{Tr}_+\left(f_\alpha(\widehat\pi)^{1/2} A_{F,+} f_\alpha(\widehat\pi)^{1/2} \right).\label{omegaTrace}
\ee
For $\alpha = 1$, the functional $\omega_\alpha$ reduces to the tracial state $\mathrm{Tr}_+$. 
For $0<\alpha<1$, $f_\alpha(\widehat{\pi})$ is an unbounded operator and therefore is not an element of $\mathcal{M}_+$. 
To make sense \eqref{omegaTrace}, we approximate the unbounded operator $f_\alpha(\widehat{\pi})$ with a family of bounded operators by introducing a spectral cutoff: for any $n>0$, set $f_\alpha^{(n)}(\widehat{\pi}) = f_\alpha(\widehat{\pi})\,\chi_{[0,n]}(\widehat{\pi}) \in \mathcal{M}_+$. The functional $\omega_\alpha$ is then defined as the limit $\omega_\alpha(A_{F,+}) = \lim_{n \to \infty} \mathrm{Tr}_+\big(f_\alpha^{(n)}(\widehat{\pi})^{1/2} A_{F,+}\,f_\alpha^{(n)}(\widehat{\pi})^{1/2}\big)$.

The evaluation $\omega_\alpha(A_{F,+})$ admits the following integral representation:
\be
\o_\a\lt( A_{F,+}\rt)
&=&\int^{\infty}_{0}\rmd\pi\, e^{-\pi}f_\a\left(\pi\right)\,\wt F(\pi)\langle\Psi_{0}\mid  a \mid\Psi_{0}\rangle =\a \int^{\infty}_{0}\rmd\pi\, e^{-\alpha\pi}\,\wt F(\pi)\langle\Psi_{0}\mid  a \mid\Psi_{0}\rangle\nonumber\\
&=& \langle\Psi_{0}\otimes\zeta_\a\mid A_{F,+}\mid\Psi_{0}\otimes\zeta_\a\rangle,\label{omegaIntegral}
\ee
where $\zeta_\a(\pi)=\sqrt{\alpha}e^{-\alpha\pi/2}$ is the normalized wave function of the clock (up to a phase). The inner product $\langle \cdot\mid\cdot\rangle$ is defined on the Hilbert space $\ch_{\rm tot}$. The state $\Psi_0\otimes\zeta_\a\in\ch_{\rm tot}$ is the purification of the state $\o_\a$.
Recall that operators of the form $\wh a_F$ are dense in $\mathcal{N}$. Given any $\wh a \in \mathcal{N}$, there exists a net $\{\wh a_F\}$ such that $a_F$ converges to $\wh a$ in the weak operator topology. By the result \eqref{omegaIntegral}, it follows that $\o_\a\lt( \Pi_+\wh a_F\Pi_+\rt)$ converges to $\o_\a\lt( \Pi_+\wh a\Pi_+\rt)$ as $a_F \to a$. Therefore, the definition of $\o_\a$ in \eqref{omegaIntegral} extends by continuity to all of $\mathcal{M}_+$:
\be
\o_\a\lt( \Pi_+\wh a\, \Pi_+\rt)=\langle\Psi_{0}\otimes\zeta_\a\mid \Pi_+\wh a\,\Pi_+\mid\Psi_{0}\otimes\zeta_\a\rangle,\qquad \forall\, \wh a\in\cn,
\ee
Thus $\o_\a$ is manifestly a positive and normalized state on $\cm_+$, and $\o_\a$ is continuous in the weak operator topology. In addition, the state $\o_\a$ is faithful, i.e. $\o_\a(\wh a^*\wh a)=0$ implies $\wh a=0$, for $\wh a\in\cm_+$. See Appendix \ref{Faithfulness of omega} for a proof.

Comparing to $\mathrm{Tr}_+(\wh a_F)$ in \eqref{PiaFPi}, the weight $f_\alpha$ in this integral modifies the clock energy distribution to $e^{-\pi}f_\alpha(\pi)=\alpha e^{-\alpha\pi}$.
In the limit $\alpha \to 0$, states with large clock energy are no longer suppressed, in contrast to the tracial state.

The behavior that the density state $f_\a(\wh \pi)$ enhances the heavy-clock contributions can also be understood by studying the limit of $f_\a$ as $\a\to0$. Specifically, for any fixed $c\geq0$,
\be
f_\alpha\left(c\log\frac{1}{\alpha}\right)
=\alpha^{1-c+c\alpha}
\xrightarrow[\alpha\to0]{}
\begin{cases}
0,&0\leq c<1,\\
1,&c=1,\\
\infty,&c>1.
\end{cases}
\ee
The crossover of $f_\alpha(\pi)$ from suppression to enhancement occurs at $\pi\sim\log(1/\alpha)$. Moreover, for every fixed cutoff $\Lambda<\infty$,
\be
\omega_\alpha\left(\chi_{[0,\Lambda]}(\widehat\pi)\right)
=\a\int_0^\Lambda\rmd\pi\,e^{-\a\pi}
=1-e^{-\alpha\Lambda}
\xrightarrow[\alpha\to0]{}0.
\ee
As $\alpha \to 0$, restricting to states with small clock energy yields a negligible contribution, 
so nontrivial results arise only when contributions from states with large clock energy are included. 
Due to these properties, we refer to $\omega_\alpha$ for small $\alpha$ as the heavy-clock state, and $\alpha \to 0$ as the heavy-clock limit.

The discussion can be generalized to a broader class of states:
Let $g(x)$ be any (complex-valued) function supported on $\R_{\geq 0}$ that satisfies the normalization $\int_0^\infty |g(x)|^2\rmd x=1$. We define
\be
f_\a(\wh\pi)=\alpha e^{\wh \pi} |g(\a \wh \pi)|^2.\label{generalizefalpha}
\ee
With this generalized $f_\a(\wh\pi)$, the state $\o_\a$ define in \eqref{omegaTrace} satisfies the positivity, normalization, and continuity in the same way as the above. In addition, $\o_\a$ is faithful if $f_\alpha(\widehat{\pi})^{-1/2} $ (with restriction of spectrum $\pi\geq 0$) is a bounded operator, see Appendix \ref{Faithfulness of omega} for details. 
In analog with \eqref{omegaIntegral}, the purification is given by $\Psi_0\otimes \zeta_\a$ where 
\be
\zeta_\a(\pi) =\alpha^{1/2} g(\a \pi), 
\ee
such that $\o_\a\lt( A_{F,+}\rt)= \langle\Psi_{0}\otimes\zeta_\a| A_{F,+}|\Psi_{0}\otimes\zeta_\a\rangle$. 
As with the original choice, this generalized $f_\a$ enhances the heavy-clock sector, because for small $\a$, the factor $e^{-\pi} f_\a(\pi) = \alpha |g(\a \pi)|^2$ in the integral \eqref{omegaIntegral} does not suppress contributions from large $\pi$.

If we further require the differentiability of function $g$ so that $\zeta_\a\in L^2(\R)$ lies in the common domain of both $\wh\tau$ and $\wh\tau^2$, then $\zeta_\a$ matches the semiclassical states considered in \cite{Chandrasekaran:2022cip}.
In this state, the quantum clock $\wh\tau$ has both expectation value and uncertainty of order $O(\alpha)$:
\be
\langle\zeta_\a|\wh\tau|\zeta_\a\rangle_{L^2(\R)}
={i}\a C_1,\qquad
\langle\zeta_\a|(\Delta\wh\tau)^2|\zeta_\a\rangle_{L^2(\R)}=-\a^2 (C_2-C_1^2)\label{semiclassicalTauC1C2}
\ee
where $C_1=\int_0^\infty\rmd x\,\overline{g(x)}g'(x)$ and $C_2=\int_0^\infty\rmd x\,\overline{g(x)}g''(x)$ are both finite. 
Under time evolution $\psi(\pi,t)=e^{-i\pi t}\zeta_\a(\pi)$, the quantum clock has the expectation value $t+O(\a)$, while the uncertainty remains $O(\a)$.
In the semiclassical limit, the quantum observer is expected to reduce to the classical observer on classical spacetime, and the observer's clock should give a well-defined time with uncertainty much smaller than the de-Sitter radius $L_{\rm dS}$ (we have set $\b_{\rm dS}=2\pi L_{\rm dS}=1$ throughout the discussion). Therefore, the limit $\a\to 0$ can be interpreted as the semiclassical limit. So we expect that in this limit, the correlation functions evaluated by $\o_\a$ should reduces to the semiclassical correlation function in the static patch.

In the following, we proceed with the general $f_\alpha(\widehat\pi)$ defined in \eqref{generalizefalpha} without any differentiability assumption. The corresponding state is still defined by $
\omega_\alpha(A_+) = \mathrm{Tr}_+\left( f_\alpha(\widehat\pi)^{1/2} A_+ f_\alpha(\widehat\pi)^{1/2} \right)$ 
for any $A_+\in\cm_+$. 

We continue to refer to the state $\o_\a$, defined using the general $f_\a(\wh \pi)$, as the heavy-clock state. We continue to refer to the $\alpha \to 0$ as the heavy-clock limit, rather than the semiclassical limit, since our discussion does not require the state to be semiclassical in the sense of \eqref{semiclassicalTauC1C2}. 
For example, the choice at the beginning $f_\alpha(\pi)=\alpha e^{(1-\alpha)\pi}$ (understood as a function supported on $\pi\geq 0$) does not correspond to a semiclassical state, since it is not differentiable at $\pi=0$, but it corresponds to a heavy-clock state.

As discussed in \cite{Chandrasekaran:2022cip}, the purified state $\Psi_0\otimes\zeta_\a\in\ch_0\otimes\ch_{\rm obs}$ can be interpreted as a physical state that is annihilated by the Hamiltonian constraint, provided that we include a second, identical observer in the complementary static patch. Here we give a brief explanation: 
By including the second observer, the Hamiltonian constraint becomes 
\be
\wh{\mathscr{H}}=\wh H+\wh \pi-\wh\pi',
\ee
which acts on the kinematical Hilbert space $\ch_0\otimes \ch_{\rm obs} \otimes \ch_{\rm obs}'$ with $\ch_{\rm obs}'\cong L^2(\R)$. Here the primed quantities are associated to the second static patch. The minus sign is because the modular time translation in the second patch is past-pointing. The gauge invariant operators in two static patches form algebras $\cn$ and $\cn'$, where $\cn'$ is generated by $e^{-is\wh\pi'}$, $s\in\R$, and $e^{i\wh\t' \wh H} a'e^{-i\wh\t' \wh H}$ for all local operators $a'\in\ca_0'$ in the second patch (here $\cn'$ does not denote the commutant of $\cn$, although $\cn'$ does belong to the commutant). 
Besides restricting to gauge-invariant operators, we must restrict the states to only physical states that are invariant under gauge transformation--that is, they are annihilated by the Hamiltonian constraint.
Typically, such states do not live in the kinematical Hilbert space, but rather are elements in the space of distributions.
Explicitly, the unitary transformation $\Fu=e^{-i\wh\tau'(\wh H+\wh \pi)}$ simplifies the constraint by
\be
\Fu\,\wh{\mathscr{H}}\,\Fu^{-1}=-\wh\pi', 
\ee
so physical states are given by\footnote{
The state relates to the coinvariant construction discussed in \cite{Chandrasekaran:2022cip} as follows. Under the unitary transformation $\Fu$, the gauge transformation is represented as $e^{is\wh\pi'}$, with $s\in\R$. We represent any $\Psi\in \ch_0\otimes \ch_{\rm obs} \otimes \ch_{\rm obs}'$ as $\Psi(\pi')$, where $\Psi(\pi')\in \ch_0\otimes \ch_{\rm obs}$ for each $\pi'\in\R$. A coinvariant is a gauge-equivalence class $\Psi\sim e^{is\wh\pi'} \Psi$ for any $s\in\R$. We define the group-averaging map $\eta$ by $\eta(\Psi):=\int \rmd s\, e^{-is \pi'}\Psi(\pi')=\Psi(0)\otimes\delta(\pi')$, which maps a dense domain $\Fd\in \ch_0\otimes \ch_{\rm obs} \otimes \ch_{\rm obs}'$ to the space $\Fd^*$ of linear functionals. The map $\eta$ depends only on the coinvariance, owing to the translational invariance of the integration measure. The image of $\eta$ consists of physical states $\Psi_{\rm phys}=\Psi(0)\otimes\delta(\pi')$, which are in one-to-one correspondence with states in $\ch_0\otimes \ch_{\rm obs}$. The inner product between physical states is given by $\langle\eta(\Psi),\eta(\Phi)\rangle_{\rm phys}:=\int\rmd \pi'\int \rmd s\, e^{is \pi'}\langle \Psi(\pi'),\Phi(\pi')\rangle=\langle \Psi(0),\Phi(0)\rangle$, showing the equivalence to the inner product on $\ch_0\otimes \ch_{\rm obs}$. Therefore, the Hilbert space spanned by $\Psi_{\rm phys}$ is isomorphic to $\ch_0\otimes \ch_{\rm obs}$. The inner product $\langle\eta(\Psi),\eta(\Phi)\rangle_{\rm phys}$ can also be viewed as defined on the space of coinvariants. Removing null vectors from the coinvariant space yields the physical Hilbert space, which is isomorphic to $\ch_0\otimes \ch_{\rm obs}$.
} 
\be
\Psi_{\rm phys}=\Phi\otimes\delta(\pi'),\qquad \forall \Phi\in \ch_0\otimes \ch_{\rm obs}, 
\ee
The physical Hilbert space is identical to $\ch_0\otimes \ch_{\rm obs}$. Since all operators in $\cn$ commute with $\Fu$, the representation of $\cn$ on $\ch_0\otimes \ch_{\rm obs}$ used above becomes the representation on the physical Hilbert space.
The representation of $\cn'$ is given by $\Fu e^{-is\wh\pi'}\Fu^{-1}\Psi_{\rm phys}=e^{-is(\wh H+\wh\pi)}\Psi_{\rm phys}$ and $\Fu e^{i\wh\t' \wh H} a'e^{-i\wh\t' \wh H}\Fu^{-1}\Psi_{\rm phys}=a'\Psi_{\rm phys}$. Next, we project the states and operators onto the subspace with $\pi>0$ and $\pi'>0$. By conjugation with $\Fu$, the projections on physical states are given by $\Fu\,\chi_+(\wh \pi)\,\Fu^{-1}\Psi_{\rm phys}=\chi_+(\wh \pi)\Psi_{\rm phys}$ and $\Fu\, \chi_+(\wh \pi')\, \Fu^{-1}\Psi_{\rm phys}=\chi_+(\wh H+\wh \pi)\Psi_{\rm phys}$. The physical Hilbert space is then restricted to $\chi_+(\wh \pi)\chi_+(\wh H+\wh \pi)(\ch_0\otimes \ch_{\rm obs})$, which contains the state $\Psi_0\otimes\zeta_\a$ used above. By applying another unitary transformation $e^{i\wh \tau \wh H}$, we obtain the isomorphism between the physical Hilbert space and $\chi_+(\wh \pi-\wh{H})\chi_+(\wh \pi)(\ch_0\otimes \ch_{\rm obs})$.

\section{Heavy-clock limit of n-point correlation functions}\label{sec:Holographic correspondence of npoint correlation functions}

We now turn to the computation of correlation functions in the heavy-clock state. We start with the two-point function with operator insertions $A_+ = \Pi_+ \pi_0(a) \Pi_+$ and $B_+ = \Pi_+ \pi_0(b) \Pi_+$. Recall the expression from \eqref{Tr+rhoABrho}. The two-point function in the heavy-clock state is given by
\be
\o_{\a}(A_+B_+)=\int\rmd E\,K_{\a}(E)m_{a,b}(E),
\ee
where
\be
m_{a,b}(E)=\left\langle\Psi_0\middle|a\middle|E\right\rangle\left\langle E\middle|b\middle|\Psi_0\right\rangle,\qquad 
K_\alpha(E)=\int_0^\infty\rmd\pi\,e^{-\pi}f_\alpha(\pi)\chi_+(\pi-E).
\ee

\begin{lemma}\label{lemma:ConvergenceKalpha}
For $f_\a(\wh\pi)=\alpha e^{\wh \pi} |g(\a \wh \pi)|^2$, where $g$ supported on $\R_{\geq 0}$ satisfies $\int_0^\infty |g(x)|^2\rmd x=1$, we have $0\leq K_\alpha(E)\leq 1$ and the convergence $K_\alpha(E)\to1$ for every $E\in\R$,  as $\alpha\to0$.

\end{lemma}

\begin{proof}
Let us define $r \equiv \max\{0, E\}\geq0 $. By changing variables to $x = \alpha \pi$, the integral $K_\alpha(E)$ becomes $K_\alpha(E) = \int_{\alpha r}^\infty |g(x)|^2 \, \rmd x$. 
Therefore, $0\leq K_\alpha(E) \leq  \int_0^\infty |g(x)|^2 \, \rmd x=1$ by the normalization of $g$.
As $\alpha \to 0$, the lower limit $\alpha r \to 0$, so $K_\alpha(E) \to  1$.
\end{proof}

So we recover the bulk two-point function in this limit and obtain the desired correspondence of two-point function
\be
\lim_{\alpha\to0}\omega_\alpha(A_+B_+)&=&\int\rmd E\,
\left\langle\Psi_0\middle|a\middle|E\right\rangle\left\langle E\middle|b\middle|\Psi_0\right\rangle
=\left\langle\Psi_0\middle|ab\middle|\Psi_0\right\rangle .
\ee
Note that $\left|K_\alpha(E)m_{a,b}(E)\right|\leq |m_{a,b}(E)|$, and $ |m_{a,b}(E)|$ is integrable by \eqref{schwarzinequalityab}. Then the limit can pass through the $E$ integral by the dominated convergence.

The correspondence can be generalized to $n$-point functions: Let $a_1,\ldots,a_n\in\mathcal{A}_0$ and $A_{j,+}:=\Pi_+\pi_0(a_j)\Pi_+$. We have
\be
A_{1,+}A_{2,+}\cdots A_{n,+}&=&e^{-i\widehat\tau\widehat H}\chi_+(\widehat\pi-\widehat H)\, a_1\, \chi_+(\widehat\pi-\widehat H)\, a_2\, \chi_+(\widehat\pi-\widehat H)\, \cdots \, a_n\, \chi_+(\widehat\pi-\widehat H)e^{i\widehat\tau\widehat H}.
\ee
In this case, between every adjacent pair of operator insertions $a_j$ and $a_{j+1}$ there appears a projector $\chi_+(\widehat\pi - \widehat H)$.

To evaluate this operator product in the heavy-clock state $\o_\a$, we insert resolutions of the identity for $\widehat H$ at each occurrence of the projector $\chi_+(\widehat\pi-\widehat H)$ between consecutive operators $a_j$ and $a_{j+1}$. 
It yields the following expression:
\be
\o_\a\left(A_{1,+}A_{2,+}\cdots A_{n,+}\right)
&=&\int_0^\infty\rmd\pi\,e^{-\pi}f_\a(\pi)\int\prod_{k=1}^{n-1}\rmd E_k\,
\prod_{k=1}^{n-1}\chi_+(\pi-E_k)\, m_n(E_1,\ldots,E_{n-1}).\label{nptFunction1}
\ee
where
\be
m_n(E_1,\ldots,E_{n-1})
&:=&
\left\langle\Psi_0\middle|a_1\middle|E_1\right\rangle
\left\langle E_1\middle|a_2\middle|E_2\right\rangle\cdots
\left\langle E_{n-1}\middle|a_n\middle|\Psi_0\right\rangle .
\ee
For $n \geq 3$, in order to justify the interchange of integration order, we need to require that the operators $a_1, \ldots, a_n$ are such that $m_n(E_1,\ldots,E_{n-1})$ is absolutely integrable:
\be
\int_{\mathbb R^{n-1}}\prod_{k=1}^{n-1}\rmd E_k\,
\left|m_n(E_1,\ldots,E_{n-1})\right|<\infty .
\label{nPointIntegrability}
\ee
For instance, this integrability condition is satisfied by observables that are concentrated at low energies and decay rapidly enough at high energies.

Given the condition \eqref{nPointIntegrability}, the integral in \eqref{nptFunction1} is absolutely convergent, so the order of integration can be interchanged without affecting the result. Thus, we find
\be
\o_\a\left(A_{1,+}A_{2,+}\cdots A_{n,+}\right)
&=&\int\prod_{k=1}^{n-1}\rmd E_k\,
K_\a^{(n)}(E_1,\ldots,E_{n-1})\, m_n(E_1,\ldots,E_{n-1}),\label{nptFunction2}
\ee
where the kernel $K_\a^{(n)}(E_1,\ldots,E_{n-1})$ contains all clock dependence
\be
K^{(n)}_\a(E_1,\ldots,E_{n-1})=\int_0^\infty\rmd\pi\,e^{-\pi}f_\a(\pi)\prod_{k=1}^{n-1}\chi_+(\pi-E_k).
\ee

In the integrand of $K_\a^{(n)}(E_1,\ldots,E_{n-1})$, the product of projectors requires the clock energy $\pi$ to exceed every intermediate modular energy, i.e. 
\be
\prod_{k=1}^{n-1}\chi_+(\pi-E_k)=\chi_+(\pi-E_{\max}),\qquad E_{\max}:=\max\{E_1,\ldots,E_{n-1}\}.
\ee
Then we obtain the explicit result of $K_\a^{(n)}(E_1,\ldots,E_{n-1})$, which is similar to the case of two-point function
\be
K^{(n)}_\alpha(E_1,\ldots,E_{n-1})&=&\int_0^\infty\rmd\pi\,e^{-\pi}f_\a(\pi)\chi_+(\pi-E_{\max}).
\ee
By Lemma \ref{lemma:ConvergenceKalpha}, we obtain $0\leq K^{(n)}_\alpha\leq 1$ and $K^{(n)}_\alpha(E_1,\ldots,E_{n-1})\to1$ as $\a\to0$ for every set of modular energy $E_1,\cdots,E_{n-1}$. By $\left|K^{(n)}_\alpha m_n\right|\leq\left|m_n\right|$ and the integrability condition \eqref{nPointIntegrability}, the limit $\a\to 0$ can be passed through the integration by the dominated convergence. As a result, in the heavy-clock limit, the $n$-point correlation function of dressed operators recovers the standard semiclassical correlation function in the static patch as given by quantum field theory in curved spacetime.
\be
\lim_{\alpha\to0}\omega_\alpha(A_{1,+}\cdots A_{n,+})&=&\int\prod_{k=1}^{n-1}\rmd E_k\,
\left\langle\Psi_0\middle|a_1\middle|E_1\right\rangle\left\langle E_1\middle|a_2\middle|E_2\right\rangle\cdots\left\langle E_{n-1}\middle|a_n\middle|\Psi_0\right\rangle\nonumber\\
&=&\left\langle\Psi_0\middle|a_1a_2\cdots a_n\middle|\Psi_0\right\rangle.\label{nptFunction3}
\ee

The resolution to Susskind's puzzle is simple:
since $\omega_\alpha$ is not a trace, cyclicity no longer forces the imaginary part of correlation functions to vanish, so the commutator part of correlators can survive in the holographic description.

In the above, we considered the state $\omega_\alpha(\,\cdot\,) = \langle\Psi_0 \otimes \zeta_\alpha|\,\cdot\,|\Psi_0 \otimes \zeta_\alpha\rangle$, where the cyclic and separating vector $\Psi_0$ for $\mathcal{A}_0$ is the Bunch-Davies state in the static patch. However, this construction can be generalized to other states by replacing $\Psi_0$ with a generic $\Phi\in\ch_0$, leading to a broader class of states:
\be
\o_{\Phi,\a}\lt(  \co \rt)= \langle\Phi\otimes\zeta_\a| e^{i\wh\tau\wh H}\co e^{-i\wh\tau\wh H}|\Phi\otimes\zeta_\a\rangle,\qquad \forall \co\in\cm_+.\label{omegaPhiAlpha}
\ee
For the state $\Phi\otimes\zeta_\a\in\ch_0\otimes \ch_{\rm obs} $ to make a nonzero contribution in the physical Hilbert space, it is required to be non-vanishing after projection onto $e^{i\wh\tau\wh H}\chi_+(\wh \pi)\chi_+(\wh H+\wh \pi)(\ch_0\otimes \ch_{\rm obs}) = \chi_+(\wh \pi)\chi_+(\wh \pi-\wh H)(\ch_0\otimes \ch_{\rm obs})$. The cyclicity of $\Psi_0$ means that $\ca_0|\Psi_0\rangle$ is dense in $\ch_0$. If $\Phi$ belongs to this dense domain, there is an operator $b\in\ca_0$ such that $|\Phi\rangle=b|\Psi_0\rangle$. In this case, $\o_{\Phi,\a}$ relates to $\o_{\a}$ by
\be
\o_{\Phi,\a}\lt(  \co \rt)=\o_{\a}\lt(B^{*}_{+}  \co B_{+}\rt),\qquad B_{+}=\Pi_{+}e^{-i\widehat{\tau}\widehat{H}}be^{i\widehat{\tau}\widehat{H}}\Pi_{+}.
\ee
By this relation, the correspondence \eqref{nptFunction3} trivially generalizes to $\o_\a\to\o_{\Phi,\a}$ on the left-hand side and $\Psi_0\to \Phi$ on the right-hand side.

For general situation of $\Phi\in\ch_0$ that includes $\Phi\notin\ca_0|\Psi_0\rangle$, we can proceed the derivation analogously as the above: The correlation function with respect to the new state $\Phi$ is expressed as
\be
\omega_{\Phi,\alpha}(A_{1,+}\cdots A_{n,+})=\int\prod^{n}_{j=0}\rmd E_{j}\,K^{(n+2)}_{\alpha}\left(E_{0},\cdots,E_{n}\right)m_{\Phi,n}\left(E_{0},\cdots,E_{n}\right)
\ee
where $n+1$ resolution of identities have been inserted. The function $m_{\Phi,n}\left(E_{0},\cdots,E_{n}\right)$ expressed below is again assumed to be absolutely integrable.
\be
m_{\Phi,n}\left(E_{0},\cdots,E_{n}\right)=\langle\Phi\mid E_{0}\rangle\langle E_{0}\mid a_{1}\mid E_{1}\rangle\langle E_{1}\mid a_{2}\mid E_{2}\rangle\cdots\langle E_{n-1}\mid a_{n}\mid E_{n}\rangle\langle E_{n}\mid\Phi\rangle.
\ee
By $\lim_{\a\to 0} K_\a^{(n+1)}=1$, 
we find that in the heavy-clock limit, the bulk correlation function with the state $\Phi$ is recovered:
\be
\lim_{\alpha\to0}\omega_{\Phi,\alpha}(A_{1,+}\cdots A_{n,+})=\left\langle\Phi\middle|a_1a_2\cdots a_n\middle|\Phi\right\rangle\ .\label{omegaPhiAlphaLimit}
\ee
Observe that in the limit $\a \to 0$, the state $\zeta_\a$ is supported around $\pi \sim \a^{-1}\to\infty$. So for any $\Phi$ whose support includes finite modular energies, the projection $\chi_+(\wh \pi-\wh H)\Phi\otimes\zeta_\a$ is nonzero in the physical Hilbert space.

\section{Observer with $\wh H_{\rm obs}=|\wh\pi|$, gauge invariant operators, and physical states}\label{sec:local-observer-time-reversal}

In the discussion above, we considered only the sector with $\pi\geq 0$, with the observer's Hamiltonian $\wh H_{\rm obs}=\wh\pi$. We now extend the framework to include both positive and negative momentum sectors, $\pi \geq 0$ and $\pi < 0$, by defining the observer's Hamiltonian as $\wh H_{\rm obs}=|\wh\pi|$. 
This setup is connected to the physical scenario where gravity and matter are coupled to a dust clock field, a topic to be discussed in more detail in Section \ref{sec:dust clock fields}. 
The implementation of time-reversal symmetry in this model is discussed in Section \ref{sec:Time reversal symmetry}.

Given the clock Hamiltonian $\wh H_{\mathrm{obs}}=|\wh\pi|$, the flow of $\wh\tau$ generated by $\wh H_{\mathrm{obs}}$ results in two clock orientations
\be
\wh\tau(s):=e^{\rmi t|\wh\pi|}\,\wh\tau\, e^{-\rmi t|\wh\pi|},
\qquad
\frac{\rmd\wh\tau(t)}{\rmd t}=\rmi\lt[|\wh\pi|,\wh\tau(t)\rt]=\operatorname{sgn}(\wh\pi).\label{clockOrientations}
\ee
Accordingly, we refer to the case $\pi > 0$ as the forward-running clock, and $\pi < 0$ as the backward-running clock. The associated Hilbert subspaces are denoted $\ch_\pm$, defined by
\be
\mathcal{H}_{\mathrm{obs}}=\mathcal{H}_+\oplus\mathcal{H}_-,
\qquad
\mathcal{H}_\pm:=\Pi_\pm\mathcal{H}_{\mathrm{obs}},
\ee
where $\mathcal{H}_{\mathrm{obs}}=L^2(\R)$ and $\Pi_\pm=\chi_{[0,\pm\infty)}(\wh\pi)$ are the spectral projection of $\wh\pi$ onto $\pi \geq 0$ and $\pi\leq 0$ respectively\footnote{
	On $\ch_{\rm obs}$, the inner product between two states $\langle \phi |\phi'\rangle_{\rm obs}=\int_\R\rmd\pi\, \overline{\phi(\pi)} \phi'(\pi)=\int_0^\infty\rmd\pi\, \overline{\phi(\pi)} \phi'(\pi)+\int_{-\infty}^0\rmd\pi\, \overline{\phi(\pi)} \phi'(\pi)$, while $\pi=0$ is of measure zero.
}.

By coupling the quantum clock to the system in one static patch, the total Hilbert space is $\mathcal{H}_{\mathrm{tot}}=\mathcal{H}_0\otimes\mathcal{H}_{\mathrm{obs}}$, and the total (kinematical) algebra is $\ca=\ca_0\bar\otimes\mathcal{B}(\mathcal{H}_{\mathrm{obs}})$. The Hamiltonian constraint with the new clock Hamiltonian is given by the self-adjoint operator
\be
\wh{\mathbf H}_1:=\wh H\otimes\mathbb{I}_{\rm obs}+\mathbb{I}_0\otimes|\wh\pi|.
\label{quantum_branch_constraint}
\ee
When expressing the Hilbert space as a direct sum $\ch_{\rm tot}=(\ch_0\otimes\ch_+)\oplus(\ch_0\otimes\ch_-)$, the Hamiltonian constraint takes the form as a direct sum
\be
\wh{\mathbf H}_1=\wh{\mathbf H}_{1,+}\oplus\wh{\mathbf H}_{1,-},\qquad \wh{\mathbf H}_{1,+} = \wh H \otimes \mathbb{I}_+ + \mathbb{I}_0 \otimes \wh\pi,\qquad \wh{\mathbf H}_{1,-} = \wh H \otimes \mathbb{I}_- - \mathbb{I}_0 \otimes \wh\pi.
\ee
where $\wh{\mathbf H}_{1,\pm}$ acts on $\ch_0\otimes\ch_\pm$ respectively.

In the context of (quantum) relational reference frames and their transformations \cite{Thiemann:2026vaj}, each choice of relational reference frame is characterized by selecting specific clock degrees of freedom and their conjugate momenta, along with gauge-fixing conditions. Here, the forward- and backward-oriented clocks correspond to two distinct choices of relational reference frames, associated to the gauge-fixing conditions $\tau(t) = \pm t$ (see \eqref{clockOrientations}), with $t$ a phase-space independent parameter that can be interpreted as coordinate time. The two sectors, $\mathcal{H}_0 \otimes \mathcal{H}_\pm$, can thus be viewed as arising respectively from two quantizations, each of which is within its own relational reference frame. The unitary symmetry $U$ in \eqref{Umap} below relates the quantum descriptions associated to these two frames. 
These two frames are also related by the anti-unitary time-reversal symmetry, as described in Section \ref{sec:Time reversal symmetry}.

As in the case of \eqref{HamConstr}, the Hamiltonian constraint generates gauge transformations via the flow $e^{it\wh{\mathbf H}_1}$. The corresponding algebra of gauge-invariant (relational) observables, $\cm$, is defined as the commutant of this flow within the total algebra $\ca$, i.e., it consists of all operators in $\ca$ that commute with $e^{it\wh{\mathbf H}_1}$ for all $t \in \mathbb{R}$.
\be
\mathcal{M}=
\lt\{\mathcal{O}\in\ca\;\middle|\;
e^{i t\wh{\mathbf H}_1}\mathcal{O}e^{-i t\wh{\mathbf H}_1}=\mathcal{O},
\;\forall t\in\mathbb{R}\rt\}.
\ee
The algebra $\cm$ differs from the crossed-product algebra $\cn \cong \ca_0\rtimes_\sigma \mathbb{R}$, since the Hamiltonian constraint $\wh{\mathbf{H}}_1$ governing $\cm$ is different from the constraint $\wh{\mathbf{H}}$ appearing in \eqref{HamConstr}.

In the following, we demonstrate that the algebra $\cm$ is isomorphic to $\cm_+ \otimes M_2(\mathbb{C})$, where $\cm_+$ is the algebra \eqref{cmPlus} previously constructed in the $\pi \geq 0$ sector, and $M_2(\mathbb{C})$ denotes the algebra of $2\times2$ complex matrices: Represent clock wave functions in momentum space and introduce the unitary reflection $U$ by
\be
(U\psi)(\pi)=\psi(-\pi),\qquad \forall\psi\in\ch_{\rm tot}.\label{Umap}
\ee
The operator $U$ maps $\ch_0\otimes\mathcal{H}_\pm$ onto $\ch_0\otimes\mathcal{H}_\mp$, satisfying $U\wh\pi U^{-1}=-\wh\pi,\ 
U\wh\tau U^{-1}=-\wh\tau$.
Using $U$ to identify the negative-momentum sector with the positive one, we define the unitary isomorphism $V$ by
\be
V:\mathcal{H}_{\mathrm{tot}}\longrightarrow\wt{\ch}_{\rm tot}:=
(\mathcal{H}_0\otimes\mathcal{H}_+)\otimes\mathbb{C}^2,
\qquad
V\psi=(\Pi_+\psi,\,U\Pi_-\psi)^\top, \qquad \forall\psi\in\ch_{\rm tot}.
\ee
The map $V=\mathbb{I}\oplus U$ leaves the $\ch_0\otimes\ch_+$ sector unchanged and reflects the $\ch_0\otimes\ch_-$ sector into $\ch_0\otimes\ch_+$.
Under $V$, the projections $\Pi_\pm$ become projections on $\C^2$:
\be
\wt\Pi_+:=V\Pi_+V^{-1}
=\mathbb{I}_{\mathcal{H}_0\otimes\mathcal{H}_+}\otimes
\begin{pmatrix}1&0\\0&0\end{pmatrix},
\qquad
\wt\Pi_-:=V\Pi_-V^{-1}
=\mathbb{I}_{\mathcal{H}_0\otimes\mathcal{H}_+}\otimes
\begin{pmatrix}0&0\\0&1\end{pmatrix}.
\ee

The full kinematical algebra $\ca$ represented on $\wt\ch_{\rm tot}$ is given by
\be
\wt{\mathcal{A}}:=V\mathcal{A}V^{-1}
=\mathcal{A}_0\bar\otimes\mathcal{B}(\mathcal{H}_+)
\bar\otimes M_2(\mathbb{C}).
\ee
Due to $U\wh \pi U^{-1}=-\wh \pi$ on $\ch_-$, the Hamiltonian constraint $\wh{\mathbf H}_1=\wh{\mathbf H}_{1,+}\oplus\wh{\mathbf H}_{1,-}$ takes the following tensor-product form on $\wt\ch_{\rm tot}$
\be
V\wh{\mathbf H}_1V^{-1}
=\wh{\mathbf H}_{1,+}\otimes\mathbb{I}_2,
\qquad
\wh{\mathbf H}_{1,+}:=\wh H\otimes\mathbb{I}_+
+\mathbb{I}_0\otimes\wh\pi .
\label{FoldedConstraint}
\ee
The commutant of the gauge flow on $\wt\ch_{\rm tot}$ is now transparent:
As established previously, the type II$_1$ factor $\cm_+$ forms the commutant of the gauge flow generated by $\wh{\mathbf H}_{1,+}$ inside $\mathcal{A}_0\bar\otimes\mathcal{B}(\mathcal{H}_+)$, while the full matrix algebra $M_2(\mathbb{C})$ commutes trivially with the identity $\mathbb{I}_2$.
Therefore, the algebra of gauge invariant operators in $\wt\ca$ has the following simple form
\be
\wt{\cm}
:=V\cm V^{-1}
=\mathcal{M}_+\bar\otimes M_2(\mathbb{C}),
\label{FoldedInvariantAlgebra}
\ee

We define the trace $\O$ on ${\cm}$ by 
\be
\O\lt(V^{-1}(\wh{a} \otimes m )V\rt) = \tr_+(\wh{a})  \mathrm{tr}(m),\qquad \wh{a}\in\cm_+, \qquad m\in M_2(\C),\label{traceOmega}
\ee 
where $\mathrm{tr}$ is the standard matrix trace. 
This trace $\O$ inherits faithfulness, normality, and finiteness from $\Tr_+$.

\begin{theorem}\label{thm:matrix-amplified-II1}
If $\mathcal{A}_0$ is a type III$_1$ factor, then the algebra of gauge-invariant operators $\cm$ is a type II$_1$ von Neumann factor. The trace on $\cm$ is uniquely given by $\O$ up to a complex rescaling.
\end{theorem}

\begin{proof}

The center of $\wt\cm$ is $Z(\cm_+)\bar\otimes Z\lt(M_2(\C)\rt)=\C\,\mathbb{I}_{\cm_+}\otimes\mathbb{I}_2 $, so 
$\wt\cm$ is a von Neumann factor.
By the classification of the von Neumann factors, the factors with faithful, normal, and finite trace are either type I$_n$ or type II$_1$. The algebra $\wt{\cm}$ is not isomorphic to the algebra of finite-dimensional matrices, so it cannot be type I$_n$. Therefore, the only possible type is II$_1$. Due to the unitary equivalence $\cm\cong \wt\cm$, the von Neumann algebra $\cm$ is a type II$_1$ factor. Since $\cm$ is a type II$_1$ factor, the trace on $\cm$ is unique up to rescaling.

\end{proof}

Generalizing from the discussion at the end of Section \ref{sec:two-point-crossed-product}, the physical Hilbert space can be constructed by coupling an observer with $\wh{H}_{\rm obs}'=|\wh\pi'|$ in the complementary static patch. Then the total kinematical Hilbert space is $\ch_0\otimes \ch_{\rm obs} \otimes \ch_{\rm obs}'$ with $\ch_{\rm obs}\cong\ch_{\rm obs}'\cong L^2(\R)$, and the total Hamiltonian constraint is given by
\be
\wh{\mathscr{H}}_1=\wh H+|\wh \pi|-|\wh\pi'|,\label{mathscrH1}
\ee
We define two unitary reflections $\bm{u}$ and $\bm{u}'$ respectively on $\ch_{\rm obs}$ and $\ch_{\rm obs}'$ by $\bm{u} f(\pi)=f(-\pi)$ and $\bm{u}' f'(\pi')=f'(-\pi')$. The previous operator $U$ relates to $\bm u$ by $U=\mathbb{I}_0\otimes \bm u$. By the decomposition $\ch_{\rm obs}=\ch_+\oplus\ch_-$, the previous unitary map $V$ is now expressed as $V=\mathbb{I}_0\otimes(\mathbb{I}_+\oplus \bm u)\otimes \mathbb{I}'_{\rm obs}$ with $\bm u$ restricted to the minus sector $\ch_-$. We use $\mathbb{I}'$ to denote the identity operator on $\ch_{\rm obs}'$. Analogously, using the decomposition $\ch_{\rm obs}'=\ch_+'\oplus\ch_-'$ with respect to $ \pi'\geq 0$ and $\pi'\leq 0$, we define the unitary map $V'$ by $V'=\mathbb{I}_0\otimes \mathbb{I}_{\rm obs}\otimes(\mathbb{I}\oplus \bm u')$ with $\bm u'$ restricted to $ \ch'_-$. The combined unitary map $\mathscr{V}=V'V$ gives the isomorphism
\be
\ch_0\otimes \ch_{\rm obs} \otimes \ch_{\rm obs}'\cong (\mathcal{H}_0\otimes\mathcal{H}_+\otimes \mathcal{H}_+')\otimes\mathbb{C}^2\otimes \C^2
\ee
simplifies the Hamiltonian constraint to
\be
\mathscr{V}\wh{\mathscr{H}}_1\mathscr{V}^{-1}=\wh{\mathscr{H}}_{1,+}\otimes \mathbb{I}_2\otimes \mathbb{I}_2
\ee
where $\wh{\mathscr{H}}_{1,+}$ is $\wh{\mathscr{H}}=\wh H+\wh \pi-\wh\pi'$ restricted to $\mathcal{H}_0\otimes\mathcal{H}_+\otimes \mathcal{H}_+'$. Then the argument at the end of Section \ref{sec:two-point-crossed-product} carries over and generalizes straightforwardly: now the physical Hilbert space of gauge-invariant states is isomorphic to
\be
\ch_{\rm phys}= \lt(\mathscr{P}_+\mathcal{H}_0\otimes\mathcal{H}_{\rm obs}\rt)\otimes\mathbb{C}^2\otimes \C^2, \qquad \mathscr{P}_+=\chi_+(\wh \pi)\chi_+(\wh H+\wh \pi)
\ee
The algebra $\wt\cm$ in the first static patch are represented on $\ch_{\rm phys}$ as $\mathcal{M}_+\otimes M_2(\mathbb{C})\otimes \mathbb{I}_2$, since all operators in $\wt\cm$ commute with the $\mathscr{V}\wh{\mathscr{H}}_1\mathscr{V}^{-1}$ and $\mathscr{P}_+$. The gauge invariant operators in the complementary static patch form the algebra $\wt\cm'\cong \mathcal{M}_+'\otimes  \mathbb{I}_2\otimes M_2(\mathbb{C})$, where generators in $\cm'_+$ are represented as $\mathscr{P}_+a'\mathscr{P}_+$, $\forall a'\in\ca_0'$, and $\mathscr{P}_+e^{-is(\wh H+\wh\pi)}\mathscr{P}_+$, $\forall s\in\R$, on $\mathscr{P}_+\mathcal{H}_0\otimes\mathcal{H}_{\rm obs}$.

\section{Heavy-clock, dressing map, and correlation functions}\label{sec:folded-bulk-boundary}

In the following, we focus on a single static patch. Given a bulk local field operator $a\in\ca_0$, there are two ways to dress this operator relating to two clock orientations: Relative to a forward-running clock, we have defined $A_+=\Pi_+ a(\wh\tau)\Pi_+$ as in \eqref{holographicMappingPlus}, whereas relative to a backward-running clock the same relational prescription uses $-\wh\tau$ and $\Pi_-$. We therefore have
\be
A_+:=\Pi_+ a(\wh\tau)\Pi_+
=\Pi_+e^{-\rmi\wh\tau\wh H}a e^{\rmi\wh\tau\wh H}\Pi_+,
\qquad
A_-:=\Pi_-a(-\wh\tau)\Pi_-
=\Pi_-e^{\rmi\wh\tau\wh H}a e^{-\rmi\wh\tau\wh H}\Pi_-.
\label{ForwardBackwardRelationalLifts}
\ee
They are the two relational descriptions of the same bulk operator, adapted to clocks running in opposite directions. The operators $A_\pm$ commute with the Hamiltonian constraint in their respective sectors in $\ch_{\rm tot}$:
\be
[A_+,e^{it\lt(\wh H+\wh\pi\rt)}]=0,
\qquad
[A_-,e^{it\lt(\wh H-\wh\pi\rt)}]=0,\qquad \forall t\in\R.
\ee
Including both clock orientations, a candidate dressed operator $\cw(a)$ in the algebra $\cm$ can be defined by the direct sum:
\be
\mathcal{W}(a)
:=A_+ + A_- .
\label{OrientationIndependentBulkLift}
\ee

By the isomorphism $\cm\cong \wt\cm=\cm_+\bar\otimes M_2(\C)$ established by the unitary transformation $V=1\oplus U$, the operator $\cw(a)$ can be equivalently represented as $A_+\in\cm_+$ discussed in previous sections, tensored with the identity matrix, i.e. 
\be
\wt{\cw}(a):=V\mathcal{W}(a)V^{-1}
=A_+\otimes\mathbb{I}_2 ,
\ee
by using $U\Pi_-U^{-1}=\Pi_+$ and $Ua(\wh \t)U^{-1}=a(-\wh \t)$. It is manifest that $\wt{\cw}(a)$ commutes with the Hamiltonian constraint $V\wh{\mathbf H}_1V^{-1}
=\wh{\mathbf H}_{1,+}\otimes\mathbb{I}_2$.

It is straightforward to generalize the discussion by considering a general candidate for the dressed operator $\cm_m(a)$, which replaces the identity matrix with an arbitrary matrix $m\in M_2(\C)$.
\be
\wt\cw_m(a)=A_+\otimes m,\qquad \mathcal{W}_m(a):=V^{-1}\wt{\mathcal{W}}_m(a)V.
\ee
Here the matrix $m$ may in general depend on $a$.
The operator $\mathcal{W}_m(a)$ belongs to the algebra $\cm$ of gauge invariant observables. In terms of the representation on $\ch_{\rm phys}$, there is the third tensor slot with an identity matrix, which is skipped here and below.

Any density state $\varrho(\wh\pi)$ of $\cm_+$ can be generalized to a density state $\cf_{\varrho,\rho}$ of $\cm$ by associating with any positive-definite Hermitian density matrix $\rho\in M_2(\C)$ satisfying $\mathrm{tr}(\rho)=1$:
\be
\wt\cf_{\varrho,\rho}=V\cf_{\varrho,\rho}V^{-1}=\varrho(\wh\pi)\otimes\rho
\ee
Given the bulk operators equivalent to $a_1,\cdots,a_n\in\ca_0$ and matrices $m_1,\cdots,m_n\in M_2(\C)$, 
the corresponding dressed operators $\cw_{m_1}(a_1),\cdots,\cw_{m_n}(a_n)$ are given by $\wt\cw_{m_1}(a_1)=A_{1,+}\otimes m_1,\ \cdots,\ \wt\cw_{m_n}(a_n)=A_{n,+}\otimes m_n$, where $A_{j,+}=\Pi_+e^{-\rmi\wh\tau\wh H}a_j e^{\rmi\wh\tau\wh H}\Pi_+ $, $j=1,\cdots,n$. Their correlation function with respect to $\cf_{\varrho,\rho}$ factorizes into a trace on $\cm_+$ and a trace of two-by-two matrices:
\be
\O\lt(\cf_{\varrho,\rho}^{1/2}\cw_{m_1}(a_1)\cdots \cw_{m_n}(a_n)\cf_{\varrho,\rho}^{1/2}\rt)=\tr_+\lt(\varrho(\wh\pi)^{1/2}A_{1,+}\cdots A_{n,+}\varrho(\wh\pi)^{1/2}\rt)\mathrm{tr}\lt(\rho m_1\cdots m_n\rt),
\ee
where the definition and details of the trace $\tr_+(\cdots)$ can be found in Sections \ref{sec:two-point-crossed-product} and \ref{sec:Holographic correspondence of npoint correlation functions}.

Let $\varrho(\wh\pi)$ be the heavy-clock state for $\cm_+$, i.e. $\varrho(\wh\pi)=f_\alpha(\wh{\pi})$ in \eqref{generalizefalpha} with small $\alpha>0$, restricting the spectrum of $\wh\pi$ to be non-negative.
We define the heavy-clock state $\o_{\a,\rho}$ on $\cm$ by
\be
\wt \cf_{\alpha,\rho}=V\cf_{\a,\rho}V^{-1}
=f_\alpha(\wh\pi)\otimes\rho_\a,
\qquad
\omega_{\alpha,\rho}(\co)
:=\Omega\lt(\cf_{\alpha,\rho}^{1/2}\co \cf_{\alpha,\rho}^{1/2}\rt),\qquad \forall \co\in\cm,
\label{FoldedHeavyClockState}
\ee
upon a choice of the density matrix $\rho_\a\in M_2(\C)$ for any $\a\in[0,1]$. 
In addition to $\mathrm{tr}(\rho_\a)=1$, $\rho_\a>0$ and $\rho_\a^\dagger =\rho_\a$, we assume $\rho_\a$ to be continuous as $\a\to0^+$.
The state $\omega_{\alpha,\rho}$ inherits positivity, normalization, continuity, and faithfulness from the state $\omega_\alpha$ on $\cm_+$ \footnote{
	To check positivity, diagonalize $\rho_\a$ as $ \rho_\a  = u\mathrm{diag}(\lambda_1, \lambda_2)u^\dagger$ for a unitary $u$ and eigenvalues $\lambda_1, \lambda_2 > 0$. For any $a \in \wt\cm$, define $a' = (\mathbb{I} \otimes u^\dagger)a(\mathbb{I} \otimes u)$, so that $a'$ is a $2\times2$ matrix with entries in $\cm_+$. Then,
	$
	\o_{\alpha,\rho}(a^*a) 
	= \lambda_1\,\o_\alpha(a'_{11}{}^* a'_{11} + a'_{21}{}^* a'_{21}) + \lambda_2\,\o_\alpha(a'_{12}{}^* a'_{12} + a'_{22}{}^* a'_{22})\geq 0
	$ by the positivity of $\o_\alpha$ and $\lambda_i > 0$.
	Faithfulness of $\o_\alpha$ together with $\lambda_1, \lambda_2 > 0$ ensures that if $\o_{\alpha,\rho}(a^*a)=0$, then all entries of $a'$ must vanish. This in turn implies $a=0$, and thus $\o_{\alpha,\rho}$ is faithful. 
	Since $\o_{\alpha,\rho}(a) = \lambda_1\,\o_\alpha(a'_{11}) + \lambda_2\,\o_\alpha(a'_{22})$ and $\l_1+\l_2=1$, it follows that $\o_{\alpha,\rho}$ is normalized and continuous when $\o_\alpha$ is normalized and continuous.
}.
The correlation function with respect to $\o_{\a,\rho}$ factorizes to
\be
\o_{\a,\rho}\lt(\cw_{m_1}(a_1)\cdots \cw_{m_n}(a_n)\rt)=\o_\a\lt(A_{1,+}\cdots A_{n,+}\rt)\mathrm{tr}\lt(\rho_\a m_1\cdots m_n\rt).
\ee
The state $\o_{\a,\rho}$ can be represented as a state of the physical Hilbert space $\ch_{\rm phys}$: If we denote by $\Tr_{\rm phys}$ the standard trace on $\ch_{\rm phys}$, then
\be
\o_{\a,\rho}(\co )=\Tr_{\rm phys}\lt(w_{\a,\rho}\co\rt),\qquad 
w_{\a,\rho}=\mid\Psi_0\otimes\zeta_\a\rangle\langle\Psi_0\otimes\zeta_\a\mid\otimes \rho_\a\otimes \rho',
\ee
for any density matrix $\rho'$ on the $\C^2$ factor associated to $\ch_{\rm obs}'$.

Recall the result in Section \ref{sec:Holographic correspondence of npoint correlation functions}: In the heavy-clock limit $\a\to0$, the correlation function $\o_\a\lt(A_{1,+}\cdots A_{n,+}\rt)$ for operators in $\cm_+$ recovers the bulk $n$-point correlation function given by quantum fields inside the static patch:
\be
\lim_{\alpha\to0}\omega_\alpha(A_{1,+}\cdots A_{n,+})
&=&\left\langle\Psi_0\middle|a_1a_2\cdots a_n\middle|\Psi_0\right\rangle.\label{nptFunction3}
\ee
Consequently, in the heavy-clock limit $\alpha \to 0$, the correlator $\o_{\a,\rho}\big(\cw_{m_1}(a_1) \cdots \cw_{m_n}(a_n)\big)$ recovers the bulk $n$-point function $\left\langle\Psi_0\middle|a_1a_2\cdots a_n\middle|\Psi_0\right\rangle$ if and only if the matrices $\rho_\alpha$, $m_1,\ldots,m_n$ satisfy the condition
\be
\mathrm{tr}\big(\rho_{\alpha=0} m_1 \cdots m_n\big) = 1, \qquad \forall\, n \geq 1. \label{traceCondit}
\ee

The trace condition \eqref{traceCondit}, together with the following two well-motivated assumptions, imposes strong constraints on the allowed choices of matrices $m_1,\ldots,m_n$:

\begin{description}

\item[Self-adjointness] For being physical observables, we require $\cw_m(a)$ to be self-adjoint if $a$ is self-adjoint, which in turn requires $m$ to be a normal matrix\footnote{
	On any tensor product Hilbert space $\ch_1 \otimes \ch_2$, suppose $A \otimes B = A^\dagger \otimes B^\dagger$ for nonzero operators $A$ on $\ch_1$ and $B$ on $\ch_2$. Taking matrix elements, for any $x, y \in \ch_1$ and $\phi, \psi \in \ch_2$ such that $\langle \phi, B^\dagger \psi \rangle \neq 0$, we have $\langle x, A y \rangle \langle \phi, B \psi \rangle = \langle x, A^\dagger y \rangle \langle \phi, B^\dagger \psi \rangle$. Thus, $\langle x, A^\dagger y \rangle = (\langle \phi, B \psi \rangle / \langle \phi, B^\dagger \psi \rangle) \langle x, A y \rangle$. Since this holds for all $x, y$, it follows that $A^\dagger = c A$ for some complex scalar $c = \langle \phi, B \psi \rangle / \langle \phi, B^\dagger \psi \rangle$. Similarly, $B^\dagger = d B$ for some $d \in \mathbb{C}$, and the original equality requires $c d = 1$.
}, i.e., $[m, m^\dagger] = 0$.

\item[Linear dressing map] We require the dressing map $a \mapsto \cw_m(a)$ to be well-defined as a linear map from $\ca_0$ to $\cm$. This means that the choice of matrix $m$ does not depend on the particular operator $a$ (see explanation below and Appendix \ref{app:restriction of dressing map}); in other words, $m$ is a single universal matrix applied to all operators.

\end{description}

If we write in general that $\wt\cw_m(a) = A_+(a) \otimes m(a)$, where $A_+ : \ca_0 \to \cm_+$ is linear and $m : \ca_0 \to M_2(\C)$ is a matrix-valued function, then one can show that the dressing map $\cw_m(a)$ is linear if and only if $m(a)$ is a constant matrix independent of $a$. Then the trace condition \eqref{traceCondit} reduces to $\mathrm{tr}(\rho_0 m^n)=1$ for all $n\geq 1$. Then by $[m, m^\dagger] = 0$ from the self-adjointness and $\rho_0>0$, the trace condition restricts that
\be
m=\mathbb{I}_2.
\ee
The proofs of these statements are provided in Appendix \ref{app:restriction of dressing map}.

As a result, under the requirements that the dressing map is linear and $\cw_m(a)$ entering correlation functions are self-adjoint operators, the condition 
\be
\mathrm{tr}\big(\rho_{\alpha} m_1 \cdots m_n\big) = 1, \qquad \forall\, n \geq 1.
\ee
holds for all $\a$, since all matrices $m_j$ are identities. Therefore, $\cw_{m}(a)=\cw(a)$, and we recover the semiclassical bulk correlation function in the heavy-clock limit:
\be
\lim_{\alpha\to 0}\omega_{\alpha,\rho}\big(\mathcal W(a_1) \cdots \mathcal W(a_n)\big) = \left\langle \Psi_0 \middle| a_1 a_2 \cdots a_n \middle| \Psi_0 \right\rangle.
\ee

The correspondence can be generalized beyond the Bunch-Davies state $\Psi_0$ by replacing $\Psi_0$ with a generic state $\Phi\in\ch_0$:
\be
\o_{\Phi,\a,\rho}(\co )=\Tr_{\rm phys}\lt(w_{\Phi,\a,\rho}\co\rt),\qquad 
w_{\Phi,\a,\rho}= e^{-i\wh\tau\wh H}\mid\Phi\otimes\zeta_\a\rangle\langle\Phi\otimes\zeta_\a\mid e^{i\wh\tau\wh H}\otimes \rho_\a\otimes \rho',\label{wPhiAlphaRho}
\ee
for an arbitrary normalized density matrix $\rho'$ and the standard trace $\Tr_{\rm phys}$ on the physical Hilbert space $\ch_{\rm phys}= \lt(\mathscr{P}_+\mathcal{H}_0\otimes\mathcal{H}_{\rm obs}\rt)\otimes\mathbb{C}^2\otimes \C^2,\ \mathscr{P}_+=\chi_+(\wh \pi)\chi_+(\wh H+\wh \pi)$. The state $\mid\Phi\otimes\zeta_\a\rangle$ is required to have a nonzero projection onto the range of $e^{i\wh\tau\wh H}\mathscr{P}_+e^{-i\wh\tau\wh H}=\chi_+(\wh\pi)\chi_+(\wh \pi-\wh H)$ so that $\o_{\Phi,\a,\rho}$ is non-vanishing. The result on $\Psi_0$ generalizes to
\be
\lim_{\alpha\to 0}\omega_{\Phi,\alpha,\rho}\big(\mathcal W(a_1) \cdots \mathcal W(a_n)\big) = \left\langle \Phi \middle| a_1 a_2 \cdots a_n \middle| \Phi \right\rangle.
\ee
by the trace condition and the limit \eqref{omegaPhiAlphaLimit} for $\o_{\Phi,\a}$.

\section{Time reversal symmetry}\label{sec:Time reversal symmetry}

Let us assume the bulk theory to be time-reversal invariant within one static patch. That is, we assume there exists an anti-unitary operator $\ct_0$ acting on $\ch_0$ such that (1) $\ct_0$ leaves the algebra $\ca_0$ invariant, and (2) both the modular Hamiltonian $\wh H$ and the cyclic separating state $\Psi_0$ are invariant under $\ct_0$:
\be
T_0\mathcal A_0T_0^{-1}=\mathcal A_0,\qquad 
\ct_0 f(\wh H)\ct_0^{-1}=f(\wh H), \qquad \ct_0|\Psi_0\rangle =|\Psi_0\rangle.
\ee
Here $f$ is any bounded real measurable functions on $\R$. 
In this section and Section \ref{sec:time-reversal-invariant-operators}, we focus on a single static patch with one observer.
We define the anti-unitary map $\ct_{\rm obs}$ for the observer on $\mathcal{H}_{\mathrm{obs}}$ by requiring
\be
\ct_{\rm obs} f(\wh\tau)\ct_{\rm obs}^{-1}=f(\wh\tau),
\qquad
\ct_{\rm obs} f(\wh\pi)\ct_{\rm obs}^{-1}=f(-\wh\pi).\label{Tobs}
\ee
The anti-unitary operator $\ct:=\ct_0\otimes\ct_{\rm obs}$ is the time-reversal map on $\mathcal{H}_{\mathrm{tot}}$. 
The Hamiltonian constraints in both cases of a single and two observers are invariant under time reversal: $\ct e^{it \wh{\bf H}_1}\ct^{-1}=e^{-it \wh{\bf H}_1}$ and $\ct e^{it \wh{\mathscr{H}}_1}\ct^{-1}=e^{-it \wh{\mathscr H}_1}$ for all $t\in\R$.

The anti-unitary $\ct$ exchanges the two clock-orientation sectors, i.e. $\ct$ restricts to anti-unitary isomorphisms satisfying
\be
\ct:\mathcal{H}_0\otimes\mathcal{H}_\pm\longrightarrow\mathcal{H}_0\otimes\mathcal{H}_\mp ,\qquad 
\ct(\mathbb{I}_0\otimes\Pi_\pm)\ct^{-1}=\mathbb{I}_0\otimes\Pi_\mp .
\ee
The second relation can be shown by $
\ct\chi_{[0,\pm\infty)}(\wh\pi)\ct^{-1}
=\chi_{[0,\pm\infty)}(-\wh\pi)
=\chi_{[0,\mp\infty]}(\wh\pi)
$.
If $\psi_+\in\mathcal{H}_0\otimes\mathcal{H}_+$, then $\Pi_-\ct\psi_+=\ct\Pi_+\psi_+=\ct\psi_+$, so $\ct\psi_+\in\mathcal{H}_0\otimes\mathcal{H}_-$. The reverse inclusion follows in the same way, and anti-unitarity of the restrictions is inherited from $\ct$.

The representation of time reversal on $\wt\ch_{\rm tot}$ is given by
\footnote{
	For any $(\phi_1,\phi_2)^\top\in\wt\ch_{\rm tot}$, using $V^{-1}(\phi_1,\phi_2)^\top=\phi_1+U\phi_2$, we obtain $V\ct V^{-1}(\phi_1,\phi_2)^\top =V(\ct\phi_1+\ct U\phi_2)=(\ct U\phi_2,\,U\ct\phi_1)^\top$, where $\ct\phi_1\in \ch_0\otimes\ch_-$ and $\ct U\phi_2\in \ch_0\otimes\ch_+$. We need to prove $\mathcal{T} U \phi_2 = U \mathcal{T} \phi_2$: We introduce a complex conjugation operator $K$ on $\mathcal{H}_{\mathrm{tot}} \cong L^2(\mathbb{R}, \mathcal{H}_0)$. Fix an arbitrary orthonormal basis $e_A$ for $\mathcal{H}_0$ and express $\psi(\pi)=\sum_A e_A\psi^A(\pi)$ for any $\psi\in\mathcal{H}_{\rm tot}$, we define $(K\psi)(\pi) = \sum_A e_A\overline{\psi^A(\pi)}$. $K=K^{-1}$ is anti-unitary and satisfies
	$K \widehat{\pi} K^{-1} = \widehat{\pi} \quad \text{and} \quad K \widehat{\tau} K^{-1} = -\widehat{\tau}$.
	Introducing the anti-unitary operator $X = \mathcal{T} U$, the operator $X K^{-1}$ is unitary and commutes with both $\widehat{\pi}$ and $\widehat{\tau}$.
	The Weyl algebra generated by $\wh\t,\wh\pi$ acts irreducibly on $\mathcal{H}_{\mathrm{obs}} = L^2(\mathbb{R})$. 
	Therefore, $X K^{-1} = U_S \otimes \mathbb{I}_{\mathrm{obs}}$ where $U_S$ is a unitary operator on $\mathcal{H}_0$. We have $X U = U X$, because $U K = K U$ and obviously $U$ commutes with $XK^{-1}$.
	Substituting $X = \mathcal{T} U$ and using $U^2 = \mathbb{I}$, we obtain $\mathcal{T} U = U \mathcal{T}$. As a result,
	$\widetilde{\mathcal{T}} (\phi_1, \phi_2)^\top = (U \mathcal{T} \phi_2, U \mathcal{T} \phi_1)^\top$.
}
\be
\wt\ct:=V\ct V^{-1}
=\theta\otimes\lt(\sigma_1K_{\mathbb{C}^2}\rt),\qquad \theta:=U\ct\Pi_+ ,  
\qquad
\sigma_1=\begin{pmatrix}0&1\\1&0\end{pmatrix}, 
\label{blockOffDiag}
\ee
where $K_{\mathbb{C}^2}$ denotes the complex conjugate (in the standard basis) of $\mathbb{C}^2$. Equivalently,
\be
\wt\ct(\phi_1,\phi_2)^\top
=(\theta\phi_2,\theta\phi_1)^\top ,\qquad \forall (\phi_1,\phi_2)^\top\in\wt\ch_{\rm tot},
\ee
since $(\phi_1,\phi_2)^\top=\phi_1\otimes(1,0)^\top+\phi_2\otimes(0,1)^\top$. Since $\theta$ commutes with $\mathscr{P}_+=\chi_+(\wh \pi)\chi_+(\wh H+\wh \pi)$, the time-reversal operator $\wt\ct=\theta\otimes\lt(\sigma_1K_{\mathbb{C}^2}\rt)\otimes K_{\mathbb{C}^2}$ is defined on the physical Hilbert space $\ch_{\rm phys}$, while it acts trivially on the $\C^2$ factor associated to the observer in the complementary patch.

On operators $\co\in\ca$, the time-reversal leads to the following linear anti-automorphism
\be
\Theta(\mathcal{O})=\ct\,\mathcal{O}^*\,\ct^{-1},
\ee
It satisfies $\Theta(\mathcal{O}_1\mathcal{O}_2)=\Theta(\mathcal{O}_2)\Theta(\mathcal{O}_1)$ and $\Theta(\co^*)=\Theta(\co)^*$. See Appendix \ref{anti-automorphism} for more details on anti-automorphisms. 
Since $\ct$ commutes with $\wh{\mathbf{H}}$, the anti-automorphism $\Theta$ preserves the algebra $\cm$ of gauge-invariant observables and induces an anti-automorphism on $\cm$.

The trace $\O$ on $\cm$ is time-reversal invariant:

\begin{theorem}
	The trace $\O$ is invariant under the anti-automorphism ${\Theta}$, i.e. 
	\be
	\O\lt(\Theta(\co)\rt)=\O\lt({\ct}\co^*{\ct}^{-1}\rt)=\O\lt(\co\rt),\qquad \forall \co\in\cm.\label{timeRevInvTrace}
	\ee 
\end{theorem}

\begin{proof}
The commutativity between $\ct$ and $\wh{\mathbf{H}}$ implies the commutativity between $\widetilde{\ct}$ and $\wh{\mathbf{H}}_+ \otimes \mathbb{I}_2$. Then, for any $\wh{a}\otimes m\in \wt{\cm}$, the linear map $\wt\Theta:\wh{a}\otimes m\mapsto\widetilde{\ct}(\wh{a}^*\otimes m^*)\widetilde{\ct}^{-1}$ leaves $\wt{\cm}$ invariant and is an anti-automorphism on $\wt\cm$.

If we denote the trace \eqref{traceOmega} equivalently by $\wt\O(\wh a\otimes m)$, we only need to show $\wt\O(\wt\Theta(\wh a\otimes m))=\wt\O(\wh a\otimes m)$, then \eqref{timeRevInvTrace} follows by linearity.  
Recall the block form of $\wt{\ct}$ in \eqref{blockOffDiag}. Since $\mathrm{tr}(m)= \mathrm{tr}(\sigma_1K_{\C^2} m^* K_{\C^2}^{-1}\sigma_1^{-1})$, 
the only task is to show that $\mathrm{Tr}_+(\theta\wh  a^* \theta^{-1})=\mathrm{Tr}_+(\widehat{a})$, where $\theta=U\ct \Pi_+$ and $\wh a\in\cm_+$: By the linearity and continuity of the trace, it suffices to consider the generators
$$
\wh a= e^{-i\wh\tau\wh H}\lt[\chi_{[0,\infty)}(\wh\pi-\wh H)a e^{is (\wh{H}-\wh \pi)}\chi_{[0,\infty)}(\wh\pi-\wh H)\rt]e^{i\wh\tau\wh H}\in\cm_+,\qquad a\in\ca_0.
$$  
Since $\wh{H}$ is $\ct$ invariant and $U\wh\t =-\wh\t U$, we obtain
\be
\theta\wh  a^* \theta^{-1}&=&e^{-i\wh\tau\wh H}\lt[ U\ct \chi_{[0,\infty)}(\wh\pi-\wh H)\,e^{-is \lt(\wh{H}-\wh \pi\rt)} a ^*\chi_{[0,\infty)}(\wh\pi-\wh H)\ct^{-1} U^{-1}\rt]e^{i\wh\tau\wh H}\nonumber\\
&=&e^{-i\wh\tau\wh H}\lt[\chi_{[0,\infty)}(\wh\pi-\wh H) e^{is \lt(\wh{H}-\wh \pi\rt)}\ct a^*\ct^{-1} \chi_{[0,\infty)}(\wh\pi-\wh H)\rt]e^{i\wh\tau\wh H}.\label{thetaatheta}
\ee
The time reversal $\ct=\ct_0\otimes\ct_{\rm obs}$ gives anti-automorphism on $\ca_0$, so we have that $\ct a^*\ct^{-1}=\ct_0 a^*\ct_0^{-1}\in\ca_0$ commutes with $U$ and $e^{is\wh \pi}$. We compute the trace
\be
\mathrm{Tr}_+(\theta\wh  a^* \theta^{-1})&=&\int_{0}^{\infty} \mathrm{d}\pi \, e^{-\pi} \langle \Psi_0 | \ct_0 a^* \ct_0^{-1}| \Psi_0 \rangle e^{-is\pi}
=\int_{0}^{\infty} \mathrm{d}\pi \, e^{-\pi} \langle \Psi_0 | a | \Psi_0 \rangle e^{-is\pi}\nonumber\\
&=&\int_{0}^{\infty} \mathrm{d}\pi \, e^{-\pi} \langle \Psi_0 | a \,e^{is(\widehat{H}-\pi)}| \Psi_0 \rangle =\mathrm{Tr}_+(\wh a).
\ee
We have used the anti-unitarity of $\ct_0$ and the $\ct_0$-invariance of $\Psi_0$.
\end{proof}

For the dressing map of operators, the time reversal of $\cw(a)$ relates to the time reversal of the bulk operator $a$ by
\be
\Theta\lt(\cw(a)\rt)=\cw (\vartheta_0(a)),\qquad \vartheta_0(a):=\ct_0 a^*\ct_0^{-1},\label{covarianceCondi1}
\ee
which follows by an argument similar to that used for \eqref{thetaatheta}. This relation is the time-reversal covariance of the bulk--boundary map. The dual operator $\mathcal W(a)$ is time-reversal invariant, when the bulk operator is invariant $\theta_0(a)=a$. As the generalization, the transformation of $\mathcal{W}_m(a)$ under time reversal involves the transformation of the matrix $m$
\be
\Theta\lt(\cw_m(a)\rt)=\cw_{\varsigma(m)} (\vartheta_0(a)),\qquad \varsigma(m):=\sig_1 m^\top\sig_1^{-1},\label{covarianceCondi2}
\ee
where $\varsigma$ is the time reversal anti-automorphism $\Theta$ restricted on $M_2(\C)$, recall \eqref{blockOffDiag}.

The density matrix $\rho_\alpha$ appearing in the heavy-clock state has been largely considered arbitrary. We may additionally require that the heavy-clock state $\mathcal{F}_{\alpha,\rho}$ is invariant under time-reversal symmetry. This invariance translates to the condition that $\rho_\alpha$ satisfies $\rho_\alpha = \sigma_1 \rho_\alpha^\top \sigma_1^{-1}$. Further imposing positivity, Hermiticity, and normalization $\operatorname{tr}(\rho_\alpha) = 1$ restricts $\rho_\alpha$ to depend on a single complex function of $\a$,
\be
\rho_\a=
\begin{pmatrix}
{1}/{2} & b_\a \\
\overline{b}_\a & {1}/{2}
\end{pmatrix},\qquad b_\a\in\C ,\quad |b_\a|<\frac{1}{2} .
\ee

\section{Time-reversal invariant operators}\label{sec:time-reversal-invariant-operators}

While the heavy-clock limit successfully reproduces the semiclassical correlation functions in the static patch, a recent proposal~\cite{Susskind:2026fjc} advocates an alternative approach, where one restricts attention to dressed operators that are strictly invariant under time reversal. 
This proposal relates to the interpretation of time reversal in a role similar to that of a gauge symmetry.
Although in our framework time reversal is viewed as a global symmetry that exchanges two quantum reference frames, and the dressed operators yielding the correct semiclassical limit are generally not invariant under time reversal, we will nonetheless explore the viability of the approach in which one works strictly with time-reversal invariant operators.

Let us focus on the diagonal subalgebra $\cm_{\rm diag}=\cm_+\oplus\cm_-\subset \cm$, where $\cm_\pm$ is the algebra of gauge invariant operators on $\ch_0\otimes \ch_\pm$. Given any bulk local operator $a\in\ca_0$, we define the time-reversal invariant dressing map $\eth:\ca_0\to\cm_{\rm diag}$ by
\be
\eth(a):=\Pi_+a(\wh\tau)\Pi_+
+\Pi_-a^\vartheta(-\wh\tau)\Pi_- .
\label{PiaPiaprime00}
\ee
where $a(\wh\tau)=\pi_0(a)=e^{-\rmi\wh\tau\wh H}\,a\,e^{\rmi\wh\tau\wh H}$ and $a^\vartheta\equiv \vartheta_0(a)=\ct_0 a^*\ct_0^{-1}$. The first and second terms in the definition of $\eth(a)$ belongs to $\cm_\pm$ respectively. Since $\ct\Pi_\pm\ct^{-1}=\Pi_\mp$ and 
\be
\Theta\lt(a(\wh\tau)\rt)
=\ct a(\wh\tau)^*\ct^{-1}
=e^{\rmi\wh\tau\wh H}a^\vartheta e^{-\rmi\wh\tau\wh H}
=a^\vartheta(-\wh\tau),
\ee
the second term on the right-hand side in \eqref{PiaPiaprime00} is the time reversal of the first term, therefore $\Theta(\eth(a))$ is even under time reversal: $\Theta(\eth(a))=\eth(a)$. Only when $a^\vartheta=a$, the operator $\eth(a)$ reduces to the previous dressed operator $\mathcal W(a)$. The representation of the dressed operator on $\wt\ch_{\rm tot}=(\ch_0\otimes\ch_+)\otimes\C^2$ is given by:
\be
\wt\eth(a)=V\eth(a)V^{-1}=\begin{pmatrix}
	\Pi_+a(\wh\tau)\Pi_+ & 0 \\
	0 & \Pi_+a^\vartheta(\wh\tau)\Pi_+
	\end{pmatrix}.
\ee

Instead of representing the diagonal algebra $\cm_{\rm diag}$ on $\ch_{\rm tot}=(\ch_0\otimes\ch_+)\oplus(\ch_0\otimes\ch_-)$, let us consider the representation on 
\be
\Fh_{\rm tot}:=(\ch_0\otimes\ch_+)\oplus(\ch_0\otimes\ch_-)^*
\ee
by dualizing the minus sector. The representation is given by $\varsigma=\varsigma_+\oplus\varsigma_-:
\cm_+\oplus\cm_-^{\rm op}\longrightarrow\mathcal{B}(\Fh_{\rm tot})$,
defined by
\be
\varsigma_+(X_+)|\psi_+\rangle=X_+|\psi_+\rangle,
\qquad
\varsigma_-(X_-)\langle\psi_-|=\langle\psi_-|X_- ,\qquad \forall X_\pm\in\cm_\pm,
\ee
with $|\psi_+\rangle\in\ch_0\otimes\ch_+$ and $\langle\psi_-|\in(\ch_0\otimes\ch_-)^*$. The minus branch represents the opposite algebra because
\be
\varsigma_-(X_-)\varsigma_-(Y_-)\langle\psi_-|
=\langle\psi_-|Y_-X_-
=\varsigma_-(Y_-X_-)\langle\psi_-| .
\ee
For $X=X_+\oplus X_-$ and $Y=Y_+\oplus Y_-$ belonging to $\cm_{\rm diag}$, the operator product in the representation $\varsigma$ is given by
\be
\varsigma(X)\varsigma(Y)
=X_+Y_+\oplus Y_-X_-
=\frac{1}{2}\{X,Y\}+\frac{1}{2}[X,Y]\Pi_{\rm odd},
\qquad
\Pi_{\rm odd}:=\Pi_+-\Pi_- ,
\label{varsigmaProd}
\ee
where the operator in the minus sector acts on $\langle\psi_-|\in(\ch_0\otimes\ch_-)^*$. The right-hand side of \eqref{varsigmaProd} coincides with Susskind's proposal \cite{Susskind:2026fjc} for the product of time-reversal invariant operator. Indeed, the operator product preserves the time-reversal invariance for diagonal operators: if $\Theta(X)=X$ and $\Theta(Y)=Y$, then $\varsigma(X)\varsigma(Y)$ is time-reversal invariant\footnote{
	We denote by $\iota_-$ the anti-unitary Riesz map $\iota_-:\ch_-\to\ch_-^*$, $|\psi_-\rangle\mapsto\langle\psi_-| $, and 
	we extend $\iota_-$ to $\ch_{\rm tot}$ by the identity on $\ch_+$, and we take $\Fp_+=\Pi_+$ and $\Fp_-:=\iota_-\Pi_-\iota_-^{-1}$ as the projections on $\Fh_{\rm tot}$. Time reversal then induces the unitary maps $\Ft_-:=\ct\iota_-^{-1}:\ch_-^*\to\ch_+$ and $\Ft_+:=\iota_-\ct:\ch_+\to\ch_-^* $, which satisfy $\Ft_-\Fp_-\Ft_-^{-1}=\Fp_+$ and $\Ft_+\Fp_+\Ft_+^{-1}=\Fp_-$. The representation $\varsigma$ is equivariant under time reversal $\Ft:=\Ft_+\oplus\Ft_-$: $\Ft\varsigma(X)\Ft^{-1}=\varsigma\lt(\Theta(X)\rt)$, for any $X=X_+\oplus X_-\in\cm_{\rm diag}$, by the relation $\varsigma_-(X_-)=\iota_- X_-^* \iota_-^{-1}$. If $\Theta(X)=X$ and $\Theta(Y)=Y$, then $\varsigma(X)\varsigma(Y)$ is time-reversal invariant: $\Ft\varsigma(X)\varsigma(Y)\Ft^{-1}=\varsigma(X)\varsigma(Y)$.
}.
In contrast, the operator product $XY$ on $\ch_{\rm tot}$ is not time-reversal invariant, since $\Theta$ is anti-automorphic.

Susskind in his recent works \cite{Susskind:2026fjc,IntroSusskind2603} considers the same type of observer with $H_{\rm obs}=|\wh\pi|$, which contains both forward- and backward-running clock sectors. In these papers, the operator $\Pi_{\rm odd} = \Pi_+ - \Pi_-$ between the two projectors is introduced\footnote{
	In Susskind's notation, the projectors $\Pi_+$ and $\Pi_-$ correspond to $\Pi_{\rm f}$ and $\Pi_{\rm b}$, projecting onto the forward- and backward-time clock states, respectively. The combination $\Pi_{\rm odd}=\Pi_+-\Pi_-$ matches Susskind's $\Pi_- := \Pi_{\rm f} - \Pi_{\rm b}$.
}, in order to convert the commutator of observables into a time-reversal even quantity: Let $a$ and $b$ be self-adjoint, time-reversal invariant bulk operators. Their commutator $C:=[a,b]$ is time-reversal odd $\Theta(C) = -C$. 
The operator $C$ commutates with $\Pi_{\rm odd}$, since they act on separate factors. 
Since both $C$ and $\Pi_{\rm odd}$ are time-reversal odd their product $C\Pi_{\rm odd}$ is time-reversal even: $\Theta(C \Pi_{\rm odd}) = C\Pi_{\rm odd}$.
As a result, in a time-reversal-invariant state, the separate expectations of $C$ and $\Pi_{\rm odd}$ vanish, but the expectation value of the product $C \Pi_{\rm odd}$ is not necessarily vanishing. 
The product in Eq.~\eqref{varsigmaProd} achieves the same effect: for time-reversal even operators $X$ and $Y$, the anticommutator remains time-reversal even, while the commutator is paired with $\Pi_{\rm odd}$, ensuring the overall expression is time-reversal even.

Denote the restriction of time reversal from the forward to the backward branch by
\be
\vartheta:=\Theta|_{\cm_+}:\cm_+\longrightarrow\cm_- .
\ee
For every $X\in\cm_+$, we define the map $j$ from $\cm_+$ to $\cm_+\oplus\cm_-^{\rm op}$
\be
j(X):=X\oplus\vartheta(X),
\qquad
\mathcal F:=j(\cm_+)\subset\mathcal \cm_+\oplus\cm_-^{\rm op} .
\ee
The map $j:\cm_+\to\mathcal F$ is an isomorphism\footnote{ 
	If we denote by $\star$ the product in $\cm_+\oplus\cm_-^{\rm op}$: 
	$j(X)\star j(Y)=XY\oplus\vartheta(Y)\vartheta(X)=j(XY)$, and $j(X^*)=j(X)^* $. The inverse of $j$ is the projection onto the first component. In addition, $j=\mathrm{id}\oplus \vartheta$ is continuous under strong/weak operator topology.
}. Therefore, $\mathcal F$ is a type II$_1$ factor. 
By the above argument, $\cf$ is an algebra of time-reversal invariant operators acting on $\Fh_{\rm tot}$.
It is clear that the dressing map $\eth$ relates to $j$ by 
\be
\eth(a)=j\lt(\Pi_+a(\wh\tau)\Pi_+\rt)\in \cf.
\ee

For time-reversal invariant operators, recovering the semiclassical bulk correlation function still requires application of the heavy-clock state. This is because the product of two time-reversal invariant dressed operators still involve insertions of $\Pi_\pm$, just as what has been discussed at the beginning of Section \ref{sec:two-point-crossed-product}. 
Given bulk operators $a_1,\cdots,a_n$, the product of their time-reversal invariant dressing gives
\be
V\varsigma\lt(\eth(a_1)\rt)\cdots \varsigma\lt(\eth(a_n)\rt)V^{-1}=\begin{pmatrix}
	\Pi_+a_1(\wh\tau)\Pi_+\cdots\Pi_+a_n(\wh\tau)\Pi_+ & 0 \\
	0 & \Pi_+a_n^\vartheta(\wh\tau)\Pi_+\cdots \Pi_+a_1^\vartheta(\wh\tau)\Pi_+
	\end{pmatrix}.
\ee
Then for the heavy-clock state $\o_{\a,\rho}$ with the time-reversal invariant state $\Psi_0$ and an arbitrary normalized density matrix $\rho_\a\in M_2(\C)$, the correlation function of the dressed operators is given by \footnote{
	On the backward branch, operators act from the left on the dual Hilbert space according to \(\varsigma_-(X)\langle\psi|=\langle\psi|X\), so their composition satisfies \(\varsigma_-(X)\varsigma_-(Y)=\varsigma_-(YX)\). The corresponding state is defined by \(\omega_-^{\mathrm{dual}}(\varsigma_-(X))=\omega_-(X)\), which yields \(\omega_-^{\mathrm{dual}} (\varsigma_-(X_1)\cdots\varsigma_-(X_n))=\omega_-(X_n\cdots X_1)\).
}
\be
\o_{\a,\rho}\lt(\varsigma\lt(\eth(a_1)\rt)\cdots \varsigma\lt(\eth(a_n)\rt)\rt)=\l_{1,\a}\o_{\a}\lt(\Pi_+a_1(\wh\tau)\Pi_+\cdots\Pi_+a_n(\wh\tau)\Pi_+ \rt)+\l_{2,\a}\o_{\a}\lt(\Pi_+a_n^\vartheta(\wh\tau)\Pi_+\cdots \Pi_+a_1^\vartheta(\wh\tau)\Pi_+\rt)
\ee
where $\l_{1,\a},\l_{2,\a}$, satisfying $\l_{1,\a}+\l_{2,\a}=1$, are the diagonals of  $\rho_\a$. By the heavy-clock limit, the correlation function reproduce the semiclassical bulk correlation function:
\be
\lim_{\a\to0}\o_{\a,\rho}\lt(\varsigma\lt(\eth(a_1)\rt)\cdots \varsigma\lt(\eth(a_n)\rt)\rt)&=&\l_{1,0}\left\langle\Psi_0\middle|a_1\cdots a_n\middle|\Psi_0\right\rangle+\l_{2,0}\left\langle\Psi_0\middle|a_n^\vartheta\cdots a_1^\vartheta\middle|\Psi_0\right\rangle\nonumber\\
&=&\left\langle\Psi_0\middle|a_1\cdots a_n\middle|\Psi_0\right\rangle, 
\ee
by the property of anti-automorphism $a_n^\vartheta\cdots a_1^\vartheta=(a_1\cdots a_n)^\vartheta$ and the time-reversal invariance of $\Psi_0$.

\section{$\wh{H}_{\rm obs}=|\wh\pi|$ and dust clock fields}\label{sec:dust clock fields}

We prefer the choice $\wh{H}_{\rm obs}=|\wh{\pi}|$ because it more closely mirrors the physical scenario of gravity and matter interacting with clock fields. Let us consider the clock fields to be the Brown-Kucha\v r dust, following \cite{BrownKuchar1995,Giesel2010MGIPT,Giesel:2007wn}.
The total action of the theory is
\be
S_{\rm tot}=S_{\rm EH}+S_{\rm matter}+S_{\rm dust},
\ee
which sums the Einstein-Hilbert action $S_{\rm EH}$ for gravity (with a positive cosmological constant), the action $S_{\rm matter}$ for conventional matter fields (such as scalar, fermion, and Yang-Mills), and the action $S_{\rm dust}$ for dust fields that play the role of clocks. The dust sector contains scalar fields $\rho$, $T$, $S^j$, and $W_j$ ($j = 1,2,3$): The fields $T$ and $S^j$ define the dust reference frame, where they represent the time and space coordinates respectively. $\rho$ and $W_j$ are auxiliary fields. The explicit form of the dust action $S_{\rm dust}$ is
\be
S_{\rm dust}=-\frac{1}{2}\int \rmd^4x\sqrt{|\det(g)|}\,\rho\lt[g^{\mu\nu} U_\mu U_\nu+1\rt],\qquad U_\mu=-\partial_\mu T+W_j\partial_\mu S^j.
\ee
Varying $S_{\rm dust}$ with respect to $\rho$ yields the condition $g^{\mu\nu} U_\mu U_\nu = -1$, indicating that $U^\mu$ is a unit timelike vector. Variation with respect to the metric $g_{\mu\nu}$ leads to the energy-momentum tensor of a pressureless perfect fluid, where $\rho$ is the energy density and $U^\mu$ is the fluid velocity. We require the energy density to be positive $\rho>0$. Variation with respect to $W_j$ implies that $S^j$ remains constant along the integral lines of $U^\mu$, with $T$ being the proper time along each curve. Finally, varying with respect to $T$ and $S^j$ enforces that the integral curves of $U^\mu$ are geodesics.
The fields $T$ and $S^j$ determine a dynamical foliation of the spacetime manifold $M$, with each leaf of the foliation specified by a constant value of $T$, while $S^j$ defines the spatial coordinates on each leaf.

Upon completing Dirac's algorithm (see, e.g., \cite{Giesel2010MGIPT} for details), the Hamiltonian for the theory is expressed as a linear combination of first-class and second-class constraints. We denote the conjugate momenta of $T$ and $S^j$ by $P$ and $P_j$, respectively, and the 3-metric by $q_{ab}$.
After eliminating all second-class constraints\footnote{
	The non-trivial second class constraints are $\rho^2=\frac{P^2}{{\det(q)}}\left(1+q^{ab}U_aU_b\right)^{-1}$ and $W_j=-P_j/P$.
}, the theory is left with the first-class Hamiltonian and diffeomorphism constraints:
\begin{eqnarray}
C^{\rm (tot)} &= C + \sqrt{P^2+q^{ab}C^{\rm dust}_a C^{\rm dust}_b},\qquad
C_a^{\rm (tot)} &= C_a + C^{\rm dust}_a,\qquad C^{\rm dust}_a=P\,T_{,a}+P_j\,S^j_{,a},\label{dustconstraint}
\end{eqnarray}
Here $C,C_a$ are standard contributions from gravity and matter fields.

Let us study perturbative gravity on the empty de Sitter background. We expand the metric as $g_{\mu\nu} = \mathring{g}_{\mu\nu} + \sqrt{G}\,\delta g_{\mu\nu}$, where $\mathring{g}_{\mu\nu}$ denotes the de Sitter metric and $\delta g_{\mu\nu}$ is the perturbation. We expand the following constraint in a power series in $G$ and consider the limit $G\to 0$:
\be
\mathbf{H}(\xi)=\int_\Sig\rmd^3 \theta\lt(\xi_\perp C^{\rm (tot)}+\xi^aC_a^{\rm (tot)}\rt),\qquad \forall \xi^\mu=\xi_\perp n^\mu+\xi^a q^\mu_a
\ee
where $\Sig$ is the spatial slice (with the unit normal $n^\mu$) of the de Sitter global coordinates $(\ft,\vec{\theta})$. 
Non-perturbatively, the gravity contribution to $\mathbf{H}(\xi)$ has the prefactor $G^{-1}$, so the leading order in small $G$ is linear in $\delta g_{\mu\nu}$:
\be
\mathbf{H}(\xi)=\frac{1}{\sqrt{G}}\int_\Sig\rmd^3 \theta\lt(\delta p_{ab}\cl_\xi\mathring{q}_{ab}-\delta q_{ab}\cl_\xi\mathring{p}_{ab}\rt)+O(G^0)
\ee
where $p^{ab}$ is the momentum conjugate to $q_{ab}$. 
For arbitrary $\xi^\mu$, quantizing this leading-order term leads to a gauge constraint imposed on the physical state~\cite{Higuchi:1991tk}.

Let us focus on the case where $\xi^\mu$ is a Killing vector of de Sitter spacetime. Since $\cl_\xi\mathring{q}_{ab}=\cl_\xi\mathring{p}_{ab}=0$, the $O(G^{-1/2})$ term in the constraint $\mathbf{H}(\xi)$ vanishes identically, so the leading-order constraint is of $O(G^0)$ and has a gravitational contribution quadratic in the metric perturbation.
Let us denote the $O(G^0)$ contribution by $\mathbf{H}^{(0)}(\xi)$. In the limit $G \to 0$, we have $\mathbf{H}(\xi) \to \mathbf{H}^{(0)}(\xi)$, 
where gravity and matter decouple.
The gravity sector of $\mathbf{H}^{(0)}(\xi)$ has already been analyzed in \cite{Higuchi:1991tk}, and the matter sectors including the dust are their Hamiltonians on the fixed de Sitter background.  

Let us focus on the dust sector and transform the total canonical action of dust into the co-moving frame of the dust \cite{Giesel2010MGIPT}. Let $t$ be any time coordinate whose $(\partial/\partial t)^\mu=N n^\mu+q^\mu_a N^a$ is not necessarily orthogonal to $\Sig$. 
The dust fields are functions $T(t,\theta^a)$ and $S^j(t,\theta^a)$. We define the fields $\tilde{S}^{a}(t,\sigma)$ and $\tilde{T}(t,\sigma)$ by 
\be
S^{j}\left(t,\tilde{S}^{a}(t,\sigma)\right)=\sigma^j,\qquad T\left(t,\tilde{S}^{a}(t,\sigma)\right)=\tilde{T}(t,\sigma).
\ee
Here, $\tilde{S}^{a}(t,\sigma)$ is the global coordinate $\theta^a$ expressed in terms of $t$ and the (dynamical) dust coordinate $\sig^j$.
The hypersurface of constant $T$ is often called the dust space $\cs$, where $\sigma^{j=1,2,3}$ provides a parametrization. In terms of $\tilde{S}^a,\tilde{T}$, the action of dust within the total canonical action $\int\rmd t\rmd^3\theta \,(p_I\dot q^I-NC^{\rm tot}-N^a C_a^{\rm tot})$ is expressed as
\be
\int \rmd t\int_\cs \rmd^3 \sig\lt(\tilde P\partial_{t}\tilde T+\tilde P_a \partial_{t} \tilde{S}^a+N^a\tilde P_a-N \sqrt{\tilde{P}^2+{q}^{ab}\tilde{P}_{a}\tilde{P}_{b}}\rt).
\ee
where all the tilded quantities are fields transformed onto the dust space, in particular,
\be
\tilde{P}=\frac{P}{\left|\det\left(\partial S/\partial \tilde S\right)\right|},\qquad 
\tilde{P}_{a}=-\frac{P\partial_{a}T+P_{j}\partial_{a}S^{j}}{\left|\det\left(\partial S/\partial \tilde S\right)\right|}
\ee
This action contains no spatial derivatives with respect to $\sigma$. Variation with respect to $T$ and $\tilde S^a$ shows that both $\tilde P$ and $\tilde P_a$ are conserved quantities. If we choose the initial condition $\tilde P_a=0$, the dust sector of the Hamiltonian reduces to $\int_\cs\rmd^3\sig\, N|\tilde P|=\int_\Sig\rmd^3 \theta\, N | P|$. Identifying $N$ and $\xi_\perp$, we obtain
\be
\mathbf{H}^{(0)}(\xi)=H+\int_\Sig\rmd^3 \theta\, \xi_\perp | P|.
\ee
The Hamiltonian $H$ generates the Killing flow acting on the fields: $\{\phi,H\}=\cl_\xi \phi$, where $\phi$ contains the on-shell graviton and conventional matter fields.
Consider the timelike Killing vector $\xi$ generating the modular flow in the static patch. For this choice, the quantization of $H$ gives the modular Hamiltonian, and the function $\xi_\perp$ is given by $\xi_\perp = \cos \theta_1$, where $\theta_1 \in [0, \pi]$.

If all fields are supported within an infinitesimal neighborhood $\cu\subset\cs$, the canonical action of dust approximates a worldline action
\be
\int \rmd t\lt(\pi \dot\t + p_a\dot \theta^a+{N}^a p_a-N\sqrt{\pi ^2+{{q}^{ab}{p}_{a}{p}_{b}}}\rt)
\ee
by identifying $\t=\tilde T$, $\pi=\int_{\cu}\rmd^3\sigma\,\tilde P$, $\theta^a=\tilde{S}^a$, and $p_a=\int_{\cu}\rmd^3\sigma\,\tilde P_a$. 
This canonical action is equivalent to that of a relativistic particle with mass $\pm\pi$ and proper time parameter $\pm\t$ \cite{Giesel2010MGIPT}, which can be interpreted as describing the dynamics of an observer.
Again by the conservation of $p_a$ and the initial condition $p_a=0$, we obtain
\be
\mathbf{H}^{(0)}(\xi)=H+\cos \theta_1 | \pi|,
\ee
for $\xi$ generating the modular flow in the static patch. If we choose $t=\ft$ to be the time coordinate of the global patch, we have $N=1$ and $N^a=0$, and the variation with respect to $p_a$ gives $\dot\theta^a\propto q^{ab}p_a=0$, in other words, the worldline is a timelike geodesic located at constant $\theta^a$. 
The only such worldlines lying completely within a single static patch are those at $\theta_1 = 0$ or $\theta_1 = 180^\circ$, corresponding respectively to $r=0$ in the left and right static patches. Thus, we have
\be
\mathbf{H}^{(0)}(\xi)=\begin{cases}
	H+| \pi|, & \text{when the worldline is at $r=0$ in the left static patch},\\
    H-| \pi|, & \text{when the worldline is at $r=0$ in the right static patch}.
\end{cases}
\ee
The quantization of the first case reproduces exactly the previous constraint $\wh{\bf H}_1$ with $\wh H_{\rm obs}=|\wh\pi|$. If we include both worldlines and denote the canonical variable as primed for the worldline in the right patch, we obtain 
\be
\mathbf{H}^{(0)}(\xi)=H+| \pi|-| \pi'|,
\ee
whose quantization reproduces the constraint $\wh{\mathscr{H}}_1$ in \eqref{mathscrH1}.

Next, let us analyze the action of time reversal on the fields. 
Consider the foliation with $\Sig$ with coordinates $\theta^a$ and arbitrarily time coordinate $t$. In the coordinates, the spacetime metric takes the form $\rmd s^2=-N^2dt^2+q_{ab}(\rmd\theta^a+N^a\rmd t)(\rmd\theta^b+N^b\rmd t)$. The time reversal diffeomorphism is 
\begin{equation}
(t,\theta^a)\mapsto (t',\theta'^a)=(-t,\theta^a).
\end{equation}
The new time function $t'$ and the associated vector field $\partial/\partial{t'}$ defines the opposite time-orientation. In the decomposition $(\partial/\partial{t'})^\mu = N'{n'}^\mu + q'{}^\mu_\a {N'}^a$, we set the lapse $N'$ to be positive as for $N$, then ${n'}^\mu=-n^\mu$ be the unit normal of $\Sig$ aligned with the time-orientation defined by ${t'}^\mu$. Rewriting the same line element in the primed coordinates indicates the following transformation of fields
\begin{equation}
N'(t',\theta')=N(t,\theta),\qquad N'^a(t',\theta')=-N^a(t,\theta),\qquad q'_{ab}(t',\theta')=q_{ab}(t,\theta).
\end{equation}
By these transformations, the extrinsic curvature transforms as $K'_{ab}=-K_{ab}$ and therefore $p_{ab}\to -p_{ab}$. For the dust, $\rho$, $W_j$, $T$ and $S^j$ are spacetime scalars invariant under time reversal, whereas the conjugate momenta flip signs:
\be
P'&=&-\frac{\rho'\sqrt{\det(q')}}{N'}\left(U'_{t'}-N'^aU'_a\right)
=\frac{\rho\sqrt{\det(q)}}{N}\left(U_t-N^aU_a\right)
=-P,\\
P'_j&=&-W'_jP'=W_jP=-P_j,
\ee
where $U_\mu=-\partial_\mu T+W_j\partial_\mu S^j$. Consequently, for the worldline observer, time reversal leaves $\t$ invariant and maps
\be
\pi\mapsto -\pi. 
\ee
The quantization gives the time-reversal acting on the operators \eqref{Tobs}.

\begin{acknowledgements}

The authors are indebted to Zhenbin Yang for numerous illuminating discussions. In particular, this work receives valuable contributions from him on the idea of heavy-clock state, time-reversal invariant operators, and the \(G\to0\) limit of gravitational constraints. This work was made possible through the support of the WOST, WithOut SpaceTime project (https://withoutspacetime.org), supported by Grant ID\# 63683 from the John Templeton Foundation (JTF). The opinions expressed in this work are those of the author(s) and do not necessarily reflect the views of the John Templeton Foundation. In addition, MH receives support from the Center for SpaceTime and the Quantum (CSTQ) and from the National Science Foundation through grant PHY-2512890.
\end{acknowledgements}

\appendix

\section{Takesaki duality}\label{Takesaki duality}

Let the total Hilbert space be $\mathcal{H}_{\text{tot}} = \mathcal{H}_0 \otimes L^2(\mathbb{R}, \mathrm{d}\pi)$, and let the total algebra be $\mathcal{A} = \mathcal{A}_0 \bar{\otimes} \mathcal{B}(L^2(\mathbb{R}, \mathrm{d}\pi))$. The constraint operator is $\widehat{\mathbf{H}} = \widehat{H} \otimes \mathbb{I}_{\text{obs}} + \mathbb{I}_S \otimes \widehat{\pi}$.
We introduce the Fourier-transform unitary $\mathcal{F}: L^2(\mathbb{R}, \mathrm{d}\pi) \to L^2(\mathbb{R}, \mathrm{d}\tau)$ defined by
$$
(\mathcal{F}\psi)(\tau) = \frac{1}{\sqrt{2\pi}} \int_{\mathbb{R}} {e^{i\pi\tau}} \psi(\pi) \, \mathrm{d}\pi .
$$
Let $U_{\mathcal{F}} = \mathbb{I}_S \otimes \mathcal{F}$. This identifies $\mathcal{H}_{\text{tot}}$ with $L^2(\mathbb{R}_\tau, \mathcal{H}_0) \cong \mathcal{H}_0 \otimes L^2(\mathbb{R}_\tau)$.
Under this transformation, the multiplication operator $\widehat{\pi}$ is mapped to the generator of translations,
$$
(\l_t f)(\tau)= (e^{-it\widehat{\pi}} f)(\tau)= f(\tau - t),
$$
and the unitary group generated by the constraint $\widehat{\mathbf{H}}$ transforms as
$$
\tilde{W}(t) := U_{\mathcal{F}} e^{it\widehat{\mathbf{H}}} U_{\mathcal{F}}^* = e^{it\widehat{H}} \otimes \l_{-t}.
$$

The gauge-invariant algebra is defined as the commutant of the flow $\tilde{W}(t)$ in the algebra $\tilde{\mathcal{A}} = \mathcal{A}_0 \bar{\otimes} \mathcal{B}(L^2(\mathbb{R}_\tau))$:
$$
\tilde{\mathcal{A}}^{\widehat{\mathbf{H}}} = \{ \co \in \tilde{\mathcal{A}} \mid \tilde{W}(t) \co \tilde{W}(t)^* = \co, \quad \forall t \in \mathbb{R} \}.
$$

The total von Neumann algebra is $\mathcal{A} = \mathcal{A}_0 \bar{\otimes} \mathcal{B}(L^2(\mathbb{R}))$. We seek the gauge-invariant observable algebra
$$
\mathcal{A}^{\widehat{\mathbf{H}}} = \mathcal{A} \cap \{e^{it\widehat{\mathbf{H}}}\}'.
$$
By the bicommutant theorem, taking the commutant gives\footnote{Given von Neumann algebras $\cm_1,\cm_2$, the join $\mathcal{M} \vee \mathcal{N}=(\mathcal{M} \cup\mathcal{N})''$ is the smallest von Neumann algebra that contains both $\mathcal{M}$ and $\mathcal{N}$. The commutant of an intersection is the join of the commutants, $(\cm_1\cap\cm_2)'=\cm_1'\vee\cm_2'$.}
\be
\mathcal{A}^{\widehat{\mathbf{H}}\prime} = \mathcal{A}' \vee \{e^{it\widehat{\mathbf{H}}}\}'' = (\mathcal{A}_0' \otimes \mathbb{I}) \vee \{e^{it\widehat{H}} \otimes \lambda_{-t}\}''.\label{joinAlge}
\ee

There are two equivalent descriptions of the crossed product $\mathcal{N} = \mathcal{A}_0 \rtimes_\sigma \mathbb{R}$ acting on $L^2(\mathbb{R}_\t, \mathcal{H}_0)$:

\begin{enumerate}
\item $\mathcal{N}$ is generated by $ \lambda_t =e^{-it\wh{\pi}}$ and the operators $\pi_0(a)=e^{-i\wh\t \widehat{H}}a e^{i\wh\t\widehat{H}}$ for $a \in \mathcal{A}_0$, given explicitly by
\be
(\lambda_t \xi)(\t)=\xi(\t-t),\qquad (\pi_0(a)\xi)(\t) = \sigma_{-\t}(a)\xi(\t) = e^{-i\t \widehat{H}}a e^{i\t\widehat{H}}\xi(\t),\qquad \forall \xi\in L^2(\R_\t,\mathcal{H}_0).\label{crossprodGenerator}
\ee

\item Equivalently, $\mathcal{N}$ is generated by elements $a \in \mathcal{A}_0$ and operators of the form $e^{is\widehat{H}}\otimes \lambda_s$. This description is unitarily equivalent to the one above in \eqref{crossprodGenerator}. Indeed, defining the unitary map $U = e^{i\widehat{\tau} \widehat{H}}$, we have $U^* a U = \pi_0(a)$ and
\be
	U^* (e^{is\widehat{H}}\otimes \lambda_s) U = e^{-i\widehat{\tau} \widehat{H}} e^{is\widehat{H}} e^{i(\widehat{\tau} - s)\widehat{H}} \lambda_s = \lambda_s.
\ee
Both definitions therefore lead to isomorphic crossed-product algebras under this unitary transformation. We denote by $\mathcal{M}_1 = U\mathcal{N}U^*$ the von Neumann algebra generated by $\mathcal{A}_0 \otimes \mathbb{I}$ and $e^{it\widehat{H}} \otimes \lambda_t$.

\end{enumerate}
Our goal is to prove $\mathcal{A}^{\widehat{\mathbf{H}}\prime} = \mathcal{N}'$, from which $\mathcal{A}^{\widehat{\mathbf{H}}} = \mathcal{N}$ follows directly.

Next, we conjugate the generators of $\mathcal{A}^{\widehat{\mathbf{H}}\prime}$ in \eqref{joinAlge}. For $x' \in \mathcal{A}_0'$,
\be
(U(x' \otimes \mathbb{I})U^*\xi)(\tau) = e^{i\tau\widehat{H}}x'e^{-i\tau\widehat{H}}\xi(\tau) = \sigma'_\tau(x')\xi(\tau),\label{multiplicationOperator}
\ee
where $\sigma'_\t(x') = e^{i\t \widehat{H}}x'e^{-i\t \widehat{H}}$ is the modular automorphism leaving $\mathcal{A}_0'$ invariant. We denote this multiplication operator in \eqref{multiplicationOperator} by $\pi'_{0,\text{right}}(x')=e^{i\wh \t \widehat{H}}x'e^{-i\wh\t \widehat{H}}$.

For the translation-involving generator in \eqref{joinAlge},
$$
\lt(U(e^{it\widehat{H}} \otimes \lambda_{-t})U^*\xi\rt)(\t) = e^{i\t\widehat{H}} e^{it\widehat{H}} e^{-i(\t+t)\widehat{H}} \xi(\t+t) = \xi(\t+t),
$$
so $U(e^{it\widehat{H}} \otimes \lambda_{-t})U^* = \mathbb{I} \otimes \lambda_{-t}$. Thus $\mathcal{M}_2 = U\mathcal{A}^{\widehat{\mathbf{H}}\prime}U^*$ is generated by $\pi'_{0,\text{right}}(\mathcal{A}_0')$ and $\mathbb{I} \otimes \lambda_{-t}$.

We now compute the commutant $\mathcal{M}_1'$ directly and show that it equals $\mathcal{M}_2$.
Since $\mathcal{A}_0 \otimes \mathbb{I} \subseteq \mathcal{M}_1$, any operator $Y \in \mathcal{M}_1'$ must lie in $(\mathcal{A}_0 \otimes \mathbb{I})' = \mathcal{A}_0' \otimes \mathcal{B}(L^2(\mathbb{R}))$. This algebra is generated by $\mathcal{C} = L^\infty(\mathbb{R}, \mathcal{A}_0')$ and the group $\{\mathbb{I} \otimes \lambda_r\}_{r \in \mathbb{R}}$; these two types of generators are the multiplication operators and the translation operators on $L^2(\R)$.

\begin{lemma}

(1) $\mathcal{B}(L^2(\mathbb{R})) \cong L^\infty(\mathbb{R}) \rtimes \mathbb{R}$. This crossed product is generated by $ L^\infty(\mathbb{R})$ and the group $\{\mathbb{I} \otimes \lambda_r\}_{r \in \mathbb{R}}$.

(2) $\mathcal{A}_0' \bar{\otimes} \mathcal{B}(L^2(\mathbb{R}))\cong\cc \rtimes \mathbb{R}$. This crossed product is generated by $\mathcal{C} = L^\infty(\mathbb{R}, \mathcal{A}_0')$ and the group $\{\mathbb{I} \otimes \lambda_r\}_{r \in \mathbb{R}}$.

\end{lemma}

\begin{proof}

(1) Let $H = L^2(\mathbb{R})$. The von Neumann algebra $M = L^\infty(\mathbb{R})$ acts on $H$ by multiplication operators, $(M_f \eta)(t) = f(t)\eta(t)$ for $f \in L^\infty(\mathbb{R})$ and $\eta \in H$. The group $\mathbb{R}$ acts on $M$ by translation, $\alpha_r(M_f) = M_{\alpha_r(f)}$, where $\alpha_r(f)(t) = f(t-r)$.

The crossed product $\mathcal{R} = L^\infty(\mathbb{R}) \rtimes_\alpha \mathbb{R}$ is defined on the Hilbert space $L^2(\mathbb{R}, H) \cong L^2(\mathbb{R}^2)$. It is generated by the representation $\pi_\alpha(M)$ and the unitary group $u_r$ implementing the action, defined for $\xi \in L^2(\mathbb{R}^2)$ by
\be
(\pi_\alpha(f)\xi)(s, t) = (\alpha_{-s}(f)\xi(s, \cdot))(t) = f(t+s)\xi(s, t), \qquad
(u_r\xi)(s, t) = \xi(s-r, t) . \nonumber
\ee

Define a unitary operator $Z: L^2(\mathbb{R}^2) \to L^2(\mathbb{R}^2)$ by $(Z\xi)(s, t) = \xi(s+t, -s)$. Its adjoint $Z^*$ is $(Z^*\xi)(s, t) = \xi(-t, t+s)$. We compute the conjugation of the generators of $\mathcal{R}$ by $Z$:
$$ (Z \pi_\alpha(f) Z^* \xi)(s, t) = (\pi_\alpha(f) Z^* \xi)(s+t, -s) = f(-s + (s+t)) (Z^* \xi)(s+t, -s) = f(t) \xi(s, t). $$
Thus $Z \pi_\alpha(f) Z^* = 1_H \otimes M_f$, acting on $L^2(\mathbb{R}) \otimes L^2(\mathbb{R})$. Similarly, for the group representation,
$$ (Z u_r Z^* \xi)(s, t) = (u_r Z^* \xi)(s+t, -s) = (Z^* \xi)(s+t-r, -s) = \xi(-(-s), -s + (s+t-r)) = \xi(s, t-r), $$
so $Z u_r Z^* = 1_H \otimes \lambda_r$, where $\lambda_r$ is the standard translation operator on $L^2(\mathbb{R})$.

The von Neumann algebra generated by $M_f$ and $\lambda_r$ on $L^2(\mathbb{R})$ is the entire $\mathcal{B}(L^2(\mathbb{R}))$. To see this, note that their commutant is trivial: any operator commuting with all $M_f$ is a multiplication operator $M_g$, and commuting with $\lambda_r$ forces $g(t) = g(t-r)$ almost everywhere, so $g$ is constant. Therefore the transformed crossed product $Z \mathcal{R} Z^*$ is exactly $1_H \otimes \mathcal{B}(L^2(\mathbb{R}))$, where $1_H$ corresponds to the variable $s$. This establishes the isomorphism $\mathcal{R} \cong \mathcal{B}(L^2(\mathbb{R}))$.

(2) Let $\mathcal{A}_0'$ act on a Hilbert space $\mathcal{H}_0$. The algebra $L^\infty(\mathbb{R}, \mathcal{A}_0')$ is isomorphic to the von Neumann tensor product $\mathcal{A}_0' \bar{\otimes} L^\infty(\mathbb{R})$. The action of $\mathbb{R}$ on $L^\infty(\mathbb{R}, \mathcal{A}_0')$ by translation corresponds to the action $\beta = \mathrm{id} \otimes \alpha$ on $\mathcal{A}_0' \bar{\otimes} L^\infty(\mathbb{R})$.

	The crossed product $\mathcal{P} = (\mathcal{A}_0' \bar{\otimes} L^\infty(\mathbb{R})) \rtimes_\beta \mathbb{R}$ acts on $L^2(\mathbb{R}, \mathcal{H}_0 \otimes H) \cong \mathcal{H}_0 \otimes L^2(\mathbb{R}^2)$. The generators of $\mathcal{P}$ are $\pi_\beta(x \otimes f)$ for $x \in \mathcal{A}_0', f \in L^\infty(\mathbb{R})$, and the translation unitaries $v_r$. For $\xi \in \mathcal{H}_0 \otimes L^2(\mathbb{R}^2)$,
	$$ (\pi_\beta(x \otimes f)\xi)(s, t) = (\beta_{-s}(x \otimes f)\xi(s, \cdot))(t) = (x \otimes \alpha_{-s}(f))\xi(s, t) = x \xi(s, t) f(t+s). $$
	This shows $\pi_\beta(x \otimes f) = x \otimes \pi_\alpha(f)$. The group representation is $(v_r\xi)(s, t) = \xi(s-r, t)$, so $v_r = 1_{\mathcal{A}_0'} \otimes u_r$.

	Consequently, the crossed product $\mathcal{P}$ is generated by $x \otimes \pi_\alpha(f)$ and $1_{\mathcal{A}_0'} \otimes u_r$, which implies $\mathcal{P} = \mathcal{A}_0' \bar{\otimes} \mathcal{R}$. Applying the unitary transformation $1_{\mathcal{H}_0} \otimes Z$ to $\mathcal{P}$, we obtain
	$$ (1_{\mathcal{H}_0} \otimes Z) \mathcal{P} (1_{\mathcal{H}_0} \otimes Z)^* = \mathcal{A}_0' \bar{\otimes} (Z \mathcal{R} Z^*) = \mathcal{A}_0' \bar{\otimes} (1_H \otimes \mathcal{B}(L^2(\mathbb{R}))). $$
	Dropping the trivial tensor factor $1_H$ (corresponding to the dummy variable $s$), we identify $\mathcal{P}$ with $\mathcal{A}_0' \bar{\otimes} \mathcal{B}(L^2(\mathbb{R}))$ acting on $\mathcal{H}_0 \otimes L^2(\mathbb{R})$.

Under the unitary transformation $1_{\mathcal{H}_0} \otimes Z$, the generators of $\mathcal{P}$ map to the diagonal algebra $\mathcal{C} = L^\infty(\mathbb{R}, \mathcal{A}_0')$ and the group of $\lambda_r$:
\be
Z \lt[x\otimes\pi_\alpha(f)\rt] Z^* = x\otimes 1_H \otimes M_f,\qquad Z\lt[1_{\mathcal{A}'_0}\otimes u_r\rt] Z^* = 1_{\mathcal{A}_0'}\otimes 1_H \otimes \lambda_r.
\ee

\end{proof}

The condition that $Y \in \mathcal{A}_0' \bar{\otimes} \mathcal{B}(L^2(\mathbb{R}))$ commutes with the second generator $e^{it\widehat{H}} \otimes \lambda_t$ means $Y$ is fixed under the automorphism group $\gamma_t = \text{Ad}(e^{it\widehat{H}} \otimes \lambda_t)$. We evaluate $\gamma_t$ on the generators of $\mathcal{A}_0' \bar{\otimes} \mathcal{B}(L^2(\mathbb{R}))$.
For the translation $\mathbb{I} \otimes \lambda_r$,
$$
\gamma_t(\mathbb{I} \otimes \lambda_r) = (e^{it\widehat{H}} \otimes \lambda_t) (\mathbb{I} \otimes \lambda_r) (e^{-it\widehat{H}} \otimes \lambda_{-t}) = \mathbb{I} \otimes \lambda_r.
$$
For $X \in \mathcal{C}$, that is, a multiplication operator $\tau \mapsto X(\tau) \in \mathcal{A}_0'$ defined up to sets of measure zero,
$$
(\gamma_t(X)\xi)(\tau) = \lt[(e^{it\widehat{H}} \otimes \lambda_t) X (e^{-it\widehat{H}} \otimes \lambda_{-t}) \xi\rt](\tau) = e^{it\widehat{H}} X(\tau-t) e^{-it\widehat{H}} \xi(\tau) = \sigma'_t\lt(X(\tau-t)\rt)\xi(\tau).
$$

Since $\gamma_t$ fixes the translation operators, $\mathbb{I} \otimes \lambda_r \in \mathcal{M}_1'$.
Moreover, there is a dual action $\widehat{\tau}_p$ on $\mathcal{A}_0' \bar{\otimes} \mathcal{B}(L^2(\mathbb{R}))$ given by $\widehat{\tau}_p = \text{Ad}(\mathbb{I} \otimes u_p)$, where $(u_p\xi)(\tau) = e^{ip\tau}\xi(\tau)$. Evaluating the composition of their adjoint actions on the translation operators yields
$$
(\mathbb{I} \otimes u_p)(\mathbb{I} \otimes \lambda_r)(\mathbb{I} \otimes u_p^*)\xi(\tau) = e^{ip\tau } (\mathbb{I} \otimes \lambda_r) e^{-ip\tau } \xi(\tau) = e^{ip\tau} e^{-ip(\tau-r)} \xi(\tau-r) = e^{ipr} \xi(\tau-r),
$$
so $\widehat{\tau}_p(\mathbb{I} \otimes \lambda_r) = e^{ipr}(\mathbb{I} \otimes \lambda_r)$. For $X \in \mathcal{C}$, we have $\widehat{\tau}_p(X) = X$ because multiplication operators commute.
Because $\gamma_t$ leaves $X \in \mathcal{C}$ within $\mathcal{C}$ and stabilizes $\mathbb{I} \otimes \lambda_r$, it commutes with $\widehat{\tau}_p$.
Since $\mathcal{M}_1'$ is the fixed-point algebra of $\gamma_t$ in $\mathcal{A}_0' \bar{\otimes} \mathcal{B}(L^2(\mathbb{R}))$, this commutation guarantees that $\mathcal{M}_1'$ is invariant under the dual action $\widehat{\tau}_p$.

\begin{lemma}

$\mathcal{M}_1'$ is generated by $\{\mathbb{I} \otimes \lambda_r\}$ and $\mathcal{M}_1' \cap \mathcal{C}$.
\end{lemma}

\begin{proof}
The dual action $\widehat{\tau}$ is a strongly continuous homomorphism from the dual group $\widehat{\mathbb{R}} \cong \mathbb{R}$ into the automorphism group of $\mathcal{A}_0' \bar{\otimes} \mathcal{B}(L^2(\mathbb{R}))$. For each $p \in \mathbb{R}$, the automorphism $\widehat{\tau}_p$ is uniquely determined by its action on the generators:
$$
\widehat{\tau}_p(X) = X, \quad \forall X \in \mathcal{C},\qquad \widehat{\tau}_p(\lambda_r) = e^{ipr} \lambda_r, \quad \forall r \in \mathbb{R}.
$$
The fixed-point algebra of $\mathcal{A}_0' \bar{\otimes} \mathcal{B}(L^2(\mathbb{R}))$ under the dual action is $\mathcal{C} = L^\infty(\mathbb{R}, \mathcal{A}_0')$.

Let $\cm' \subseteq \mathcal{A}_0' \bar{\otimes} \mathcal{B}(L^2(\mathbb{R}))$ satisfy two properties: (1) $\lambda_r \in \cm'$ for all $r \in \mathbb{R}$, and (2) $\widehat{\tau}_p(\cm') = \cm'$ for all $p \in \mathbb{R}$ (invariance under the dual action).

Define $\cm'_0 = \cm' \cap \mathcal{C}$. Because $\mathcal{M}'_0 \subseteq \mathcal{M}'$ and $\lambda_r \in \mathcal{M}'$, the von Neumann algebra generated by $\mathcal{M}'_0$ and $\{\lambda_r\}_{r \in \mathbb{R}}$ is contained in $\mathcal{M}'$.

To prove the reverse inclusion, define the Fourier transform
$$
E_r(x) = \int_{\mathbb{R}} \widehat{\tau}_p(x\l_{-r}) \mathrm{d}p,\qquad \forall x\in \mathcal{A}_0' \bar{\otimes} \mathcal{B}(L^2(\mathbb{R})).
$$
Let $x \in \mathcal{M}'$ be an arbitrary element for which $E_r(x)$ is well-defined. Since $\mathcal{M}'$ is invariant under the dual action $\widehat{\tau}_p$, the Fourier transform $E_r(x )$ belongs to $\mathcal{M}'$. The image $E_r(x)$ always belongs to the fixed-point algebra $\mathcal{C}$, so $E_r(x)$ lies in the intersection
$$
E_r(x) \in \mathcal{M}' \cap \mathcal{C} = \mathcal{M}'_0.
$$
For a generic element $x=\int \rmd s\, X(s)\otimes\l_s$ with $X(s)\in\cc$, we obtain $E_r(x)=X(r)\otimes \mathbb{I}$ for all $r\in\R$. The element $x$ therefore lies in the algebra generated by $\mathcal{M}'_0$ and $\lambda_r$. Since this holds for a weakly dense set of elements in $\mathcal{M}'$, we obtain the reverse inclusion: $\cm'$ is contained in the algebra generated by $\mathcal{M}'_0$ and $\lambda_r$.

\end{proof}

We now find $\mathcal{M}_1' \cap \mathcal{C}$ explicitly. An element $X \in \mathcal{C}$ is in $\mathcal{M}_1'$ if and only if $\gamma_t(X) = X$ for all $t \in \mathbb{R}$. Consequently, for each $t \in \mathbb{R}$,
$$
\sigma'_t(X(\t-t)) = X(\t) \quad \text{for almost every } \t \in \mathbb{R}.
$$
Applying the automorphism $\sigma'_{-\t}$ to both sides,
$$
\sigma'_{-(\t-t)}(X(\t-t)) = \sigma'_{-\t}(X(\t)) \quad \text{for almost every } \t \in \mathbb{R}.
$$
Define $Y(\t) = \sigma'_{-\t}(X(\t)) $. The equation simplifies to
$$
Y(\t-t) = Y(\t) \quad \text{for almost every } \t \in \mathbb{R},
$$
and this holds for every $t \in \mathbb{R}$. The translation invariance of $Y \in L^\infty(\mathbb{R}, \mathcal{A}_0')$ implies that $Y$ is constant almost everywhere: there exists a fixed operator $x' \in \mathcal{A}_0'$ with $Y(\t) = x'$ for almost every $\t \in \mathbb{R}$.
Substituting back,
$$
X(\t) = \sigma'_\t(Y(\t)) = \sigma'_\t(x') \quad \text{for almost every } \t \in \mathbb{R},
$$
so $X$ is exactly the multiplication operator $\pi'_{0,\text{right}}(x')$. Therefore $\mathcal{M}_1' \cap \mathcal{C} = \pi'_{0,\text{right}}(\mathcal{A}_0')$.

We conclude that $\mathcal{M}_1'$ is generated by $\pi'_{0,\text{right}}(\mathcal{A}_0')$ and $\mathbb{I} \otimes \lambda_r$. Comparing with the definition of $\cm_2$, we obtain
$$
\mathcal{M}_1' = \mathcal{M}_2.
$$
Recall that $\mathcal{M}_1 = U\mathcal{N}U^*$, $\cn=\mathcal{A}_0 \rtimes_\sigma \mathbb{R}$, and $\mathcal{M}_2 = U\mathcal{A}^{\widehat{\mathbf{H}}\prime}U^*$. Conjugating this equality back by $U$ and taking the commutant of both sides recovers the Takesaki duality relation:
$$
\mathcal{A}^{\widehat{\mathbf{H}}} =  \mathcal{A}_0 \rtimes_\sigma \mathbb{R}.
$$

\section{Trace on the crossed-product algebra}\label{Trace on cross-product algebra}

Let $\mathcal{A}_0$ be a von Neumann factor acting on a Hilbert space $\mathcal{H}_0$, with a cyclic and separating\footnote{A state $\psi\in{\cal H}_S$ is a cyclic vector for a von Neumann algebra ${\cal A}_0$ if ${\cal A}_0\psi$ is dense in ${\cal H}_S$. A state $\psi\in{\cal H}_S$ is separating for ${\cal A}_0$ if the condition $a\psi=0$, $a\in{\cal A}_0$, implies $a=0$.} state $|\Psi_0\rangle$. Let $\Delta = e^{-\widehat{H}}$ be the corresponding modular operator on $\mathcal{H}_0$, where $\widehat{H}$ is the modular Hamiltonian. The modular automorphism group is $\sigma_t(a) = \Delta^{-it} a \Delta^{it} = e^{it\widehat{H}} a e^{-it\widehat{H}}$ for all $a \in \mathcal{A}_0$.

We extend the Hilbert space to $\mathcal{H}_{\rm tot} = \mathcal{H}_0 \otimes L^2(\mathbb{R})$. Let $X=-\wh{\pi}$ be the coordinate operator on $L^2(\mathbb{R})$ and $P = -i\partial_X$ its conjugate momentum.
The crossed-product von Neumann algebra $\mathcal{N} = \mathcal{A}_0 \rtimes_\sigma \mathbb{R}$ is the algebra acting on $\mathcal{H}$ generated by (1) the operators $a \otimes 1$ for all $a \in \mathcal{A}_0$, and (2) the unitary group $U(s) = e^{is(\widehat{H} + X)}$ for $s \in \mathbb{R}$.

A trace $\mathrm{Tr}$ on a von Neumann algebra $\mathcal{A}$ is a map from the cone of positive operators $\ca^+$ to the extended non-negative real numbers $\mathbb{R}_{\geq 0}$, required to satisfy the following properties \cite{Sorce:2023fdx}:
\begin{enumerate}
    \item \textbf{Positive homogeneity:} For all $T \in \ca^+$ and all $\lambda \geq 0$, $\mathrm{Tr}(\lambda T)=\lambda\, \mathrm{Tr}(T)$.
    \item \textbf{Additivity:} For all $T, S \in \ca^+$, $\mathrm{Tr}(T+S)=\mathrm{Tr}(T)+\mathrm{Tr}(S)$.
    \item \textbf{Unitary invariance:} For all $T \in \ca^+$ and all unitary $U \in \mathcal{A}$, $\mathrm{Tr}(U T U^*)=\mathrm{Tr}(T)$.
\end{enumerate}
Given the first two properties, property 3 is equivalent to the condition that for any $O \in \mathcal{A}$, $\tr(O O^*) = \tr(O^* O)$. The subset $\mathcal{A}_1 \subseteq \mathcal{A}$ of operators $T$ for which $\mathrm{Tr}(|T|)<\infty$ is a two-sided ideal in $\mathcal{A}$, the trace-class operators. The trace $\mathrm{Tr}$ then extends uniquely to a linear functional $\mathrm{Tr}: \mathcal{A}_1 \rightarrow \mathbb{C}$. The extended trace is cyclic on $\mathcal{A}_1$: for all $T \in \mathcal{A}_1$ and $S \in \mathcal{A}$, $\tr(TS)=\tr(ST)$.

A convenient basis for the crossed-product algebra $\mathcal{N}$ consists of operators of the form $\widehat{a} = a\, e^{is(\widehat{H} + X)}$, where $a \in \mathcal{A}_0$.
The trace is defined by \cite{Witten:2021unn}
\be
\mathrm{Tr}(\widehat{a}) = \lag\Psi_0\otimes f_{\rm obs}\mid \widehat{a}K^{-1}\mid \Psi_0\otimes f_{\rm obs}\rag =\int_{-\infty}^{\infty} \mathrm{d}X \, e^X \langle \Psi_0 | a \,e^{is(\widehat{H} + X)}| \Psi_0 \rangle,\label{traceD1}
\ee
where $K=e^{-(\wh{H}+X)}f_{\rm obs}(\wh{H}+X)^2$. The state $f_{\rm obs} \in \mathcal{H}_{\mathrm{obs}}$ is arbitrary, subject only to $f_{\rm obs}(X) > 0$.
The trace $\tr(\cdot)$ does not depend on the choice of $f_{\rm obs}$.

We demonstrate the cyclicity of the trace $\mathrm{Tr}(\widehat{a}\widehat{b}) = \mathrm{Tr}(\widehat{b}\widehat{a})$.
A dense subspace of $\cn$ is spanned by the operators
\be
\widehat{a}(X) = \int \mathrm{d}s \, a(s) e^{is(\widehat{H} + X)},\qquad a(s)=a f(s), \qquad a\in \mathcal{A}_0,\label{repDenseOp}
\ee
where $f(s)$ is a finite linear combination of Hermite functions. For two such operators $\widehat{a}$ and $\widehat{b}$, we evaluate the trace of $\wh{a}\wh{b}$:
\be
\mathrm{Tr}(\widehat{a}\widehat{b}) &=&  \int\rmd X\int \mathrm{d}s \mathrm{d}r \, e^{i(s+r-i)X} \langle \Psi_0 \mid a(s) e^{is\widehat{H}} b(r) e^{ir\widehat{H}} \mid \Psi_0 \rangle \nonumber\\
&=&\int \mathrm{d}r \,\langle \Psi_0 \mid \sigma_{r-i}\lt(a(-r+i)\rt) b(r)\mid \Psi_0 \rangle,
\ee
using $\wh{H} \mid\Psi\rangle = 0$ and the analyticity of Hermite functions.

We compute $\mathrm{Tr}(\widehat{b}\widehat{a})$ in the same way:
\be
\mathrm{Tr}(\widehat{b}\widehat{a}) = \int\rmd X\int \mathrm{d}s \mathrm{d}r \, e^{i(s+r-i)X} \langle \Psi_0 \mid b(r) \sigma_r\lt( a(s) \rt)\mid \Psi_0 \rangle .
\ee
We now invoke the KMS property of the modular automorphism group, which states that for any $a,b \in \mathcal{A}_0$, $ \langle \Psi _0 \mid b \sigma_{r}(a) \mid \Psi_0 \rangle = \langle \Psi_0 \mid \sigma_{r-i}(a)b \mid \Psi_0 \rangle $. The trace $\mathrm{Tr}(\widehat{b}\widehat{a}) $ becomes
\be
\mathrm{Tr}(\widehat{b}\widehat{a}) &=& \int\rmd X\int \mathrm{d}s \mathrm{d}r \, e^{i(s+r-i)X} \langle \Psi_0 \mid \sigma_{r-i}\lt( a(s) \rt)b(r) \mid \Psi_0 \rangle \nonumber\\
&=&\int\mathrm{d}r \, \langle \Psi_0 \mid \sigma_{r-i}\lt( a(-r+i) \rt)b(r) \mid \Psi_0 \rangle .
\ee
This matches the expression for $\mathrm{Tr}(\widehat{a}\widehat{b})$. We thus have $\mathrm{Tr}(\widehat{a}\widehat{b}) = \mathrm{Tr}(\widehat{b}\widehat{a})$, so $\mathrm{Tr}(\cdot)$ is a valid trace on the crossed-product algebra $\mathcal{N}$.

The trace $\mathrm{Tr}(\cdot)$ in \eqref{traceD1} satisfies the following properties:
\begin{enumerate}
    \item \textbf{Faithfulness:} For positive operator $\wh\rho\in\cn_+$, if $\mathrm{Tr}(\wh\rho) = 0$, then $\wh\rho = 0$.
    \item \textbf{Normality:} If $\{\wh\rho_\alpha\} \subset \mathcal{N}_{+}$ is any family of positive operators with supremum $\wh\rho = \sup_\alpha \wh\rho_\alpha$, then $\mathrm{Tr}(\wh\rho) = \sup_\alpha \mathrm{Tr}(\wh\rho_\alpha)$.
    \item \textbf{Semi-finiteness:} The ideal of elements $y \in \mathcal{N}$ such that $\mathrm{Tr}(y^* y) < \infty$ is weakly dense in $\mathcal{N}$.
\end{enumerate}

We sketch the proof.
For faithfulness, suppose $\widehat{\rho} \in \mathcal{N}$ is positive with $\mathrm{Tr}(\widehat{\rho}) = 0$. This implies $\langle \Psi_0 | \widehat{\rho} (X)| \Psi_0 \rangle = 0$, which is possible only if $\widehat{\rho}(X)|\Psi_0\rangle = 0$ for almost every $X \in \mathbb{R}$. Since $|\Psi_0\rangle$ is separating for $\mathcal{A}_0$, it follows that $\widehat{\rho}(X) = 0$ for almost every $X \in \mathbb{R}$, so $\widehat{\rho}=\int_\R\widehat{\rho}(X) \rmd E_X=0$ where $E_X$ is the spectrual projection. 

For normality, let $\{\widehat{\rho}_\alpha\}_{\alpha}$ be an increasing net of positive operators in $\mathcal{N}$ with supremum $\widehat{\rho}$. For any state $\eta \in \mathcal{H}_{\mathrm{tot}}$, the sequence $\langle \eta | \widehat{\rho}_\alpha | \eta \rangle$ increases monotonically and is bounded above, so it converges to a finite limit. This implies that $\langle \zeta | \widehat{\rho}_\alpha | \eta \rangle$ converges for any $\zeta, \eta \in \mathcal{H}_{\mathrm{tot}}$ to a bounded sesquilinear form $B(\zeta, \eta)$, with $B(\zeta, \eta) = \langle \zeta | \widehat{\rho} | \eta \rangle$. Thus $\widehat{\rho}_\alpha$ converges to $\widehat{\rho}$ in the weak operator topology. The normality of the trace then follows from the continuity of the functional $\langle \Psi_0 \otimes f_{\rm obs} | \cdot K^{-1}| \Psi_0 \otimes f_{\rm obs} \rangle$ with respect to the weak operator topology.

For semi-finiteness, consider the dense subspace of $\mathcal{N}$ consisting of operators of the form \eqref{repDenseOp}, where $f(s)$ is rapidly decaying and its Fourier transform $\hat{f}(X)$ satisfies $\int \mathrm{d}X\, e^{X} |\hat{f}(X)|^2 < \infty$. For such operators, the trace $\Tr(\wh{a}^*\wh{a})$ evaluates to
\be
\Tr(\wh{a}^*\wh{a})=\int\rmd X e^{X}|\hat{f}(X)|^2\lag\Psi_0|a^*a|\Psi_0\rag<\infty.
\ee

\section{Faithfulness of $\o_\a$}\label{Faithfulness of omega}

We show that for any positive element $\widehat{\mathcal{O}} =\wh a^*\wh a$ with $\wh a\in \mathcal{M}_+$, $\omega_\alpha(\widehat{\mathcal{O}}) = 0$ implies $ \widehat{a} = 0$, for $0<\a\leq1$.  

Using the spectral cutoff $f^{(n)}_\a(\wh \pi) = f_\alpha(\widehat{\pi}) \chi_{[0,n]}(\wh \pi)$ for any $n>0$, we define $\o_{\a,n}= \mathrm{Tr}_+\big(f_\alpha^{(n)}(\widehat{\pi})^{1/2} \,\cdot\,f_\alpha^{(n)}(\widehat{\pi})^{1/2}\big)$, then
\be
\omega_\alpha(\widehat{\mathcal{O}}) - \omega_{\a,n}(\widehat{\mathcal{O}}) = \int^{\infty}_{0}\rmd\pi\,e^{-\pi}\lt[f_{\alpha}(\pi)-f_{\alpha}^{(n)}(\pi)\rt]\left\langle \Psi_{0}\left|\wh \co(\pi)\right|\Psi_{0}\right\rangle 
\ge 0\ .
\ee
So $\omega_\alpha(\widehat{\mathcal{O}}) = 0$ implies $\omega_{\a,n}(\widehat{\mathcal{O}})= 0$. Since $f_\alpha^{(n)}(\widehat{\pi})^{1/2}\wh \co \,f_\alpha^{(n)}(\widehat{\pi})^{1/2}$ is a positive operator in $\mathcal{M}_+$ and $\mathrm{Tr}_+$ is a faithful trace, this implies $\widehat{a} f_\alpha^{(n)}(\widehat{\pi})^{1/2} = 0$ for all $n>0$.

The inverse $f_\alpha(\widehat{\pi})^{-1/2} $ (with restriction of spectrum $\pi\geq 0$) is a bounded operator in $\mathcal{M}_+$, for $0<\a\leq 1$. Multiplying $\widehat{a} f_\alpha^{(n)}(\widehat{\pi})^{1/2} = 0$ on the right by $f_\alpha(\widehat{\pi})^{-1/2} $ gives $\widehat{a}\chi_{[0,n]}(\wh \pi)=0$ for all $n>0$. Under the strong limit $n\to\infty $, the projection $\chi_{[0,n]}(\wh \pi)$ converges to $\Pi_+$, so we obtain $\widehat{a}\Pi_+=\wh a=0$. This proves that $\omega_\alpha$ is a faithful state.

\section{Restriction of dressing map}\label{app:restriction of dressing map}

Given that $\wt\cw_m(a) = A_+(a) \otimes m(a)$, where $A_+ : \ca_0 \to \cm_+$ is a linear map and $m : \ca_0 \to M_2(\C)$ is a (matrix-valued) function. Let us first prove that the map $\wt\cw_m : \ca_0 \to \wt\cm$ is linear if and only if $m(a)$ is a constant matrix independent of $a$:

It suffices to prove the "only if" direction, as the converse is obvious.
Let $a_0 \in \ca_0$ be an element such that $A_+(a_0) \neq 0$. We define the constant matrix $m_0 = m(a_0)$. Let $a \in \ca_0$ be another arbitrary element. We will show that $m(a)=m_0$ by considering two cases based on whether $A_+(a)$ and $A_+(a_0)$ are linearly independent in $\cm_+$.

First, let $A_+(a)$ and $A_+(a_0)$ be linearly independent\footnote{{The image of $A_+$ is infinite-dimensional because the dressing map is injective. To prove this, let $P_p:=\chi_{(-\infty,p]}(\wh H)=
	\chi_{[0,\infty)}(p-\widehat H)$ be the spectral projection of the Hamiltonian. By Eq.~\eqref{exchangeRelationPiPlus}, conjugating $A_+(a)$ by $e^{\rmi\wh\tau\wh H}$ gives an operator $P_paP_p$, where $p$ is the spectrum of $\wh\pi$. Hence $A_+(a)=0$ implies $P_paP_p=0$ for almost every $p\in\R$. As $p\to\infty$, $P_{p_k}\to\mathbb{I}_{\ch_0}$ strongly. Since $a\in\ca_0$ is bounded, we obtain the strong limit $P_{p}aP_{p}\to a$, so $a=0$. Thus $A_+$ is injective. Since $\ca_0$ is a type III$_1$ factor and therefore infinite-dimensional, it follows that $\dim(\mathrm{im}A_+)=\infty$.}
}. We evaluate $\wt\cw_m(a+a_0)$ by the linearity of both $A_+$ and $\wt\cw_m$,
\be
A_+(a) \otimes m(a + a_0) + A_+(a_0) \otimes m(a + a_0) = A_+(a) \otimes m(a) + A_+(a_0) \otimes m_0 
\ee
which implies $m(a + a_0)=m(a)=m_0$ by the linear independence of $A_+(a)$ and $A_+(a_0)$.
	
Consider the case where $A_+(a)$ and $A_+(a_0)$ are linearly dependent, i.e., $A_+(a)=\l A_+(a_0)$, with $\l\in\C$. There exists an element $b \in \ca_0$ such that $A_+(b)$ and $A_+(a_0)$ are linearly independent in $\cm_+$. The above argument gives $m(b) = m_0$. Moreover, since $A_+(a)=\l A_+(a_0)$, the operators $A_+(a)$ and $A_+(b)$ are linearly independent if $\l\neq0$. Applying the above argument again gives $m(a)=m(b)=m_0$. If $\lambda = 0$, then $A_+(a) = 0$, so $\wt\cw_m(a) = 0$ can equivalently be written as $\wt\cw_m(a) = A_+(a) \otimes m_0$.

The trace condition \eqref{traceCondit} reduces to $\mathrm{tr}(\rho_0 m^n)=1$ for all $n\geq 1$. If $\cw_m(a)$ is a self-adjoint operator, then $m$ is a normal matrix. Below, we prove that the trace condition with $\rho_0>0$ implies $m = \mathbb{I}_2$:

Any 2-by-2 matrix $m$ satisfies
$
m^2 - \mathrm{tr}(m)\,m + \det(m)\,\mathbb{I}_2 = 0.
$
Multiplying by $\rho_0$ and taking the trace, we deduce $1 - \mathrm{tr}(m) + \det(m) = 0$ by using $\mathrm{tr}(\rho_0) =\mathrm{tr}(\rho m) = \mathrm{tr}(\rho_0 m^2) = 1$.
Any $2 \times 2$ matrix $u$ satisfies $\det(u - \mathbb{I}_2) = \det(u) - \mathrm{Tr}(u) + 1$, so we obtain $\det(m - \mathbb{I}_2) = 0$.

Define $B:=m-\mathbb{I}_2$. Since $\det(B) = 0$, $B$ has rank at most one. If $B\neq 0$, it can be written as $B = u v^\dagger$ for nonzero $u, v \in \mathbb{C}^2$. Then $\mathrm{tr}(\rho B) = v^\dagger \rho u = 0$. The self-adjointness of $\cw_m(a)$ implies that $m$ is normal, so $B$ is normal: $B B^\dagger = B^\dagger B$. For $B = u v^\dagger$, this would mean $\|v\|^2 u u^\dagger = \|u\|^2 v v^\dagger$, so $u = c v$ for $c =v^\dagger u/\|v\|^2$. But then $v^\dagger \rho u = c\, v^\dagger \rho v = 0$, and the positivity of $\rho$ would imply $v = 0$ or $c=0$, either of which would imply $B = 0$, contradicting the assumption that $B$ is nonzero. Therefore, $B$ must vanish and $m = \mathbb{I}_2$.

\section{Anti-automorphism}\label{anti-automorphism}

Let $\mathcal{H}$ be a complex Hilbert space and let $\mathcal{B}(\mathcal{H})$ denote the algebra of all bounded linear operators on $\mathcal{H}$. An anti-automorphism of a $*$-algebra $\mathcal{A}$ is a bijective map $\alpha : \mathcal{A} \to \mathcal{A}$ satisfying, for all $x, y \in \mathcal{A}$ and $\lambda \in \mathbb{C}$:
\begin{itemize}
\item Linearity: $\alpha(\lambda x + y) = \lambda \alpha(x) + \alpha(y)$
\item Anti-multiplicativity: $\alpha(xy) = \alpha(y)\alpha(x)$
\item Involution preservation: $\alpha(x^*) = \alpha(x)^*$
\end{itemize}

An anti-unitary operator $J : \mathcal{H} \to \mathcal{H}$ is a surjective, conjugate-linear isometry. It satisfies:
\begin{itemize}
\item Conjugate-linearity: $J(\lambda \xi + \mu \eta) = \bar{\lambda} J \xi + \bar{\mu} J \eta$
\item Inner-product preservation: $\langle J \xi, J \eta \rangle = \langle \eta, \xi \rangle = \overline{\langle \xi, \eta \rangle}$
\end{itemize}
For an anti-unitary operator, the adjoint is defined so that $\langle J \xi, \eta \rangle = \langle J^* \eta, \xi \rangle$. It follows that $J^* J = J J^* = I$, meaning $J^* = J^{-1}$, and $J^*$ is also anti-unitary.

Let $J$ be an anti-unitary operator on $\mathcal{H}$. The map $\alpha_J : \mathcal{B}(\mathcal{H}) \to \mathcal{B}(\mathcal{H})$ defined by
$$ \alpha_J(x) = J x^* J^* $$
is an anti-automorphism of $\mathcal{B}(\mathcal{H})$.

Consider three self-adjoint operators $A,B,C\in\mathcal{A}$ satisfying $[A,B]=iC$. Their transformations under the anti-automorphism are $A \mapsto \a(A) = J A J^*$, $B\mapsto \a(B)= JBJ^*$, and $C\mapsto \a(C)= JCJ^*$. We then obtain
$$ [\a (A), \a(B)] = -i J C J^* = -i \a(C)=-\a([A,B]). $$

\bibliography{main.bib}

\end{document}